\documentclass[11pt]{article}
\usepackage[T1]{fontenc}
\usepackage{graphicx}
\usepackage{color}
\usepackage{amsmath,amssymb}
\usepackage{latexsym, amsthm}
\usepackage{enumerate}
\usepackage{bbm}
\usepackage{hyperref}
\usepackage{authblk}
\usepackage{setspace} 
\usepackage{accents}
\usepackage{ulem} 
\large\normalsize

\theoremstyle{plain}
\newtheorem{theorem}{Theorem}[section]
\newtheorem{corollary}[theorem]{Corollary}
\newtheorem{lemma}[theorem]{Lemma}
\newtheorem{proposition}[theorem]{Proposition}
\newtheorem{definition}[theorem]{Definition}
\theoremstyle{definition}
\newtheorem{remark}[theorem]{Remark}

\numberwithin{equation}{section}
\title{
\bf{Amplifying the Randomness of Weak Sources Correlated With Devices}
}

\author[1]{Hanna Wojew\'odka}
\author[2]{Fernando G. S. L. Brand\~{a}o}
\author[3]{Andrzej Grudka}
\author[4]{Karol Horodecki}
\author[5]{Micha\l{} Horodecki}
\author[6]{Pawe\l{} Horodecki}
\author[5]{Marcin Paw\l{}owski}
\author[7]{Ravishankar Ramanathan}
\author[8]{Maciej Stankiewicz}
\affil[1]{\small Institute of Mathematics, University of Silesia in Katowice, Bankowa 14, 40-007 Katowice, Poland}
\affil[2]{\small Quantum Architectures and Computation Group, Microsoft Research, Redmond, WA, USA, and also with the Department of Computer Science, University College London WC1E 6BT, UK}
\affil[3]{\small Faculty of Physics, Adam Mickiewicz University, 61-614 Pozna{\'n}, Poland}
\affil[4]{\small Institute of Informatics and the National Quantum Information Centre, Faculty of Mathematics, Physics and Informatics, University of Gda{\'n}sk, 80-309 Gda{\'n}sk, Poland}
\affil[5]{\small Institute of Theoretical Physics and Astrophysics and the National Quantum Information Centre, Faculty of Mathematics, Physics and Informatics, University of Gda{\'n}sk, 80-309 Gda{\'n}sk, Poland}
\affil[6]{\small Faculty of Applied Physics and Mathematics and the National Quantum Information Centre, Gda{\'n}sk University of Technology, 80-233 Gda{\'n}sk, Poland}
\affil[7]{\small Laboratoire d'Information Quantique, Universit{\'e} Libre de Bruxelles, Belgium}
\affil[8]{\small National Quantum Information Centre, Faculty of Mathematics, Physics and Informatics, University of Gda{\'n}sk, 80-309 Gda{\'n}sk, Poland}

\date{}
\begin{document}
\maketitle
\begin{abstract}
The problem of device-independent randomness amplification against no-signaling adversaries has so far been studied under the assumption that the weak source of randomness is uncorrelated with the (quantum) devices used in the amplification procedure. In this work, we relax this assumption, and reconsider the original protocol of Colbeck and Renner 
	using a Santha-Vazirani (SV) source. To do so, we introduce an SV-like condition for devices, namely that any string of SV source bits remains weakly random conditioned upon any other bit string from the same SV source and the outputs obtained when this further string is input into the devices. Assuming this condition, we show that a quantum device using a~singlet state to violate the chained Bell inequalities leads to full randomness in the asymptotic scenario of a~large number of settings, for a restricted set of SV sources (with $0 \leq \varepsilon < (2^{(1/12)} - 1)/(2(2^{(1/12)} + 1)) \approx 0.0144$). 
	We also study a device-independent protocol that allows for correlations between the sequence of boxes used in the protocol and the SV source bits used to choose the particular box from whose output the randomness is obtained. 
	{Assuming the SV-like condition for devices, we show that the honest parties can achieve amplification of the weak source, for the parameter range $0 \leq \varepsilon<0.0132$, against a class of attacks given as a mixture of product box sequences, made of extremal no-signaling boxes, with additional symmetry conditions}. Composable security proof against this class of attacks is provided. 
\end{abstract}

\section{Introduction}

In many applications, like numerical simulations, cryptography or gambling, just to name a~few, free randomness is desired due to the fact that a wide range of results is based on it. In practice, however, random sources are rarely private and only partially weak sources of randomness are available. That is why the problem of randomness amplification became useful and worth investigating. Overall, the idea is to use the inputs from a~partially random source and obtain perfectly random output bits. In classical information theory, randomness amplification from a~single weak source is unattainable (\cite{sv}). However, it becomes possible, if the no-signaling principle is assumed and quantum-mechanical correlations are used. Such correlations are revealed operationally through the violation of Bell inequalities.
 
As a model of a weak source to be amplified, we consider an $\varepsilon$-SV source (named after Santha and Vazirani \cite{sv}), where $\varepsilon$ is a~parameter which indicates how far we are from full randomness. {An $\varepsilon$-SV source is given by a~probability distribution $P(\varphi_1,\ldots,\varphi_n,\ldots)$ over bit strings such that
\begin{align}
	\label{SV}
	\begin{aligned}
		&(0.5-\varepsilon)\leq P(\varphi_1|e)\leq (0.5+\varepsilon),\\
		&(0.5-\varepsilon)\leq P(\varphi_{i+1}|\varphi_1,\ldots,\varphi_{i},e)\leq (0.5+\varepsilon)
	\end{aligned}
\end{align}
for every $1\leq i\leq n$, 
where $e$ represents an arbitrary random variable prior to $\varphi_1$, which can influence $\varphi_1,\ldots,\varphi_n,\ldots$. }
Note that, when $\varepsilon=0$, bits are fully random, while they can be even fully deterministic when $\varepsilon=0.5$. For brevity, throughout the rest of the paper we will write $p_-$ for $(0.5-\varepsilon)$ and $p_+$ for $(0.5+\varepsilon)$.

In the research on randomness amplification, the paper of Colbeck and Renner \cite{cr} is certainly crucial. It is also a starting point for our idea. The authors consider the bipartite scenario of the chained Bell inequality and prove that, under certain assumptions (discussed later), it is possible to amplify randomness of $\varepsilon$-SV sources, provided that $\varepsilon<\left(\sqrt{2}-1\right)^2/2\approx 0.086$. The result may be improved, as is done in \cite{ghhhpr}. There, based on the observation that extremal points of the set of probability distributions from an $\varepsilon$-SV source are certain permutations of Bernoulli distributions with parameter $(0.5-\varepsilon)$, randomness amplification is obtained for any $\varepsilon < 0.0961$. Moreover, the bound is shown to be tight, which means that under these assumptions, it is not possible to achieve randomness amplification using the chained Bell inequality above this threshold.

Gallego et al. \cite{acin} show that, given an $\varepsilon$-SV source, with any $0<\varepsilon< 0.5$, and assuming no-signaling, full randomness may be certified using quantum non-local correlations. In \cite{acin}, the Bell scenario of five-party Mermin inequality is considered, however, unlike in the protocol proposed in \cite{cr}, the hashing function used to compute the final random bit is not explicitly provided and a large number of space-like separated devices is required.

Further results were obtained in \cite{marcin}, \cite{Chung-Shi-Wu}, \cite{brghhh}, \cite{ravi}, \cite{chungNew} etc., a~wide range of protocols have been proposed, these are summarized and compared in Table I in \cite{brghhh}. The problem has been considered from different points of view and a~lot of obstacles, such as the requirement of an infinite number of devices or no tolerance for noise, have already been overcome. However, relaxing the assumption about independence between a source and a~device has not yet been widely studied, especially in the context of a finite device framework against a no-signaling adversary.  

In this paper, we relax this assumption, i.e., do not require a~source and a~device to be independent. Instead, we only limit the correlations between them by one constraint, which we call the SV-condition for boxes and specify in details later. 
We prove explicitly that the most malicious correlations (between a~source and a~device) are not allowed due to the assumption that an $\varepsilon$-SV source remains an $\varepsilon$-SV source even upon obtaining the inputs and outputs from boxes. Hence, randomness amplification is still possible. Our new method of proof allows to analyze an attack where an adversary sends to the honest parties those boxes that are particularly adapted to their measurement settings, as well as to the hashing function applied. We explain the dangers of such attacks with an explicit example in Section~\ref{sec:motivation}.

{So far, only Chung et al. 
have tried to weaken the independence assumption. In \cite{Chung-Shi-Wu} they approach the
problem in a quantum formalism, while in {\cite{chungNew}}, which was announced later than the first version of this paper, they prove (in the similar spirit) security against no-signaling adversaries, although using a larger number of devices. Our approach is different and independent from the one proposed by Chung et al. in \cite{Chung-Shi-Wu} and \cite{chungNew}. 
We believe that the results obtained within this paper give a new insight into the research on randomness amplification and, due to the clarity of assumptions, will also be significant in the more general task of obtaining secure key bits in cryptography}.

The paper is organized as follows.
In Section \ref{sec:preliminaries} we introduce some basic notations and definitions. A motivation for the paper is described in Section \ref{sec:motivation} with a~toy example of an attack strategy for the adversary.
In Section \ref{sec:theObservedBellValuefromSVcondition} we formally state the assumptions considered in the paper and discuss the results for a single no-signaling box.
Section \ref{sec:exampleTruevsObsChain} is devoted to the explicit example of the chained Bell inequality, which is interesting because it may be compared with the results of Colbeck and Renner \cite{cr}. 
Further, within Section \ref{sec:composableDistance}, we estimate a composable distance (for a~private weak source of randomness) between a fully random bit and an output bit of a box. 
In Section \ref{sec:randomnessAmplificationProtocol} we revisit the Colbeck and Renner protocol for amplification of randomness using the chained Bell inequality. 
A~general class of attacks exhibiting certain kind of symmetry and having limited correlations between the runs of the device (see Sections \ref{sseq:old:6e} and \ref{sseq:old:6f} for the detailed description of the assumptions on the attack strategy) is considered within Section \ref{sec:analysisOfProtocol}.
We prove (in a~composable way for private sources) that  under the relaxed assumption, 
against this class of attacks, the protocol allows for amplification in the parameter range $0 \leq \varepsilon<0.0132$. {Finally, in Section~\ref{sec:conclusion}, we summarize our results and raise just a~few open questions}. 


\section{Preliminaries}
\label{sec:preliminaries}

\subsection{No-signaling boxes}

In our study we use a family of probability distributions, usually called a box, denoted by $P(O|I)$, where $I$ and $O$ are random variables describing the vectors of inputs and outputs, respectively.  

To talk about randomness amplification, it is advisable to explain what is meant by the no-signaling condition. 
In the simplest case, when there are only two parties: Alice and Bob, the no-signaling assumption is that
\begingroup\makeatletter\def\f@size{9}\check@mathfonts
\begin{align}\label{eq:ns}
\begin{aligned}
\sum_{y}&P(O=(x,y)|I=(u,v))=\sum_{y}P(O=(x,y)|I=(u,v^{\prime})),
		 \\
	\sum_{x}&P(O=(x,y)|I=(u,v))=\sum_{x}P(O=(x,y)|I=(u^{\prime},v))
\end{aligned}
\end{align}
\endgroup
for every $u,u^{\prime},v,v^{\prime},x,y$, {where $u,u^{\prime}$ and $v,v^{\prime}$ denote the inputs of Alice and Bob, respectively, while $x$ and $y$ denote their outputs}.

\subsection{Bell values observed in laboratories}

Theoretically, there may exist no-signaling boxes which attain the algebraic violation of the chosen Bell inequality. However, as for now, we are able to use in laboratories only these boxes which violate the inequality up to the value obtained within the rules of quantum mechanics. This simply means that the Bell value observed in a lab may not be lower (here a larger violation is characterized by a smaller value for the Bell expression) than the value predicted by quantum mechanics.

\subsection{Bell inequalities useful for randomness amplification} 

It is well-known that quantum mechanics allows for non-local correlations between spatially separated systems. 
Occurrence of such correlations can be verified through the violation of Bell inequalities. The convex set formed by the correlations described by quantum theory is sandwiched between the sets
of classical and general no-signaling correlations. Only extremal boxes (vertices) of the no-signaling polytope are completely uncorrelated with the environment and hence provide intrinsic certified randomness. It has been recently proven in \cite{tuziemski} that non-local vertices of the no-signaling polytopes of correlations admit no quantum realization. For amplification of SV sources, Bell inequalities with the property that the optimal quantum value equals the optimal no-signaling value are required.  For such Bell inequalities (e.g. GHZ paradoxes \cite{paradox}, pseudo-telepathy games \cite{gms-pseudo-telepathy} or Bell inequalities for graph states \cite{gthb-graph-states}), or those where the quantum violation is close to algebraic (such as the chained Bell inequality \cite{braunstein_caves}), the quantum set reaches the corresponding facet of the no-signaling polytope.

In this paper we mainly focus on the chained Bell inequality, which has already been used in the research on randomness and privacy amplification (see \cite{cr}, \cite{ghhhpr} or \cite{rotem}).

\section{Motivation and a toy example}
\label{sec:motivation}

We now exemplify a possible attack that utilizes correlations between a weak source and device in the simplest scenario of boxes with binary inputs and outputs.
Even though these boxes do not constitute a resource for randomness amplification, the attack can already be described in terms of these.
	\begin{figure}[!t]
		\centering
		\includegraphics[trim=0cm 5cm 0cm 0cm, width=9.3cm]{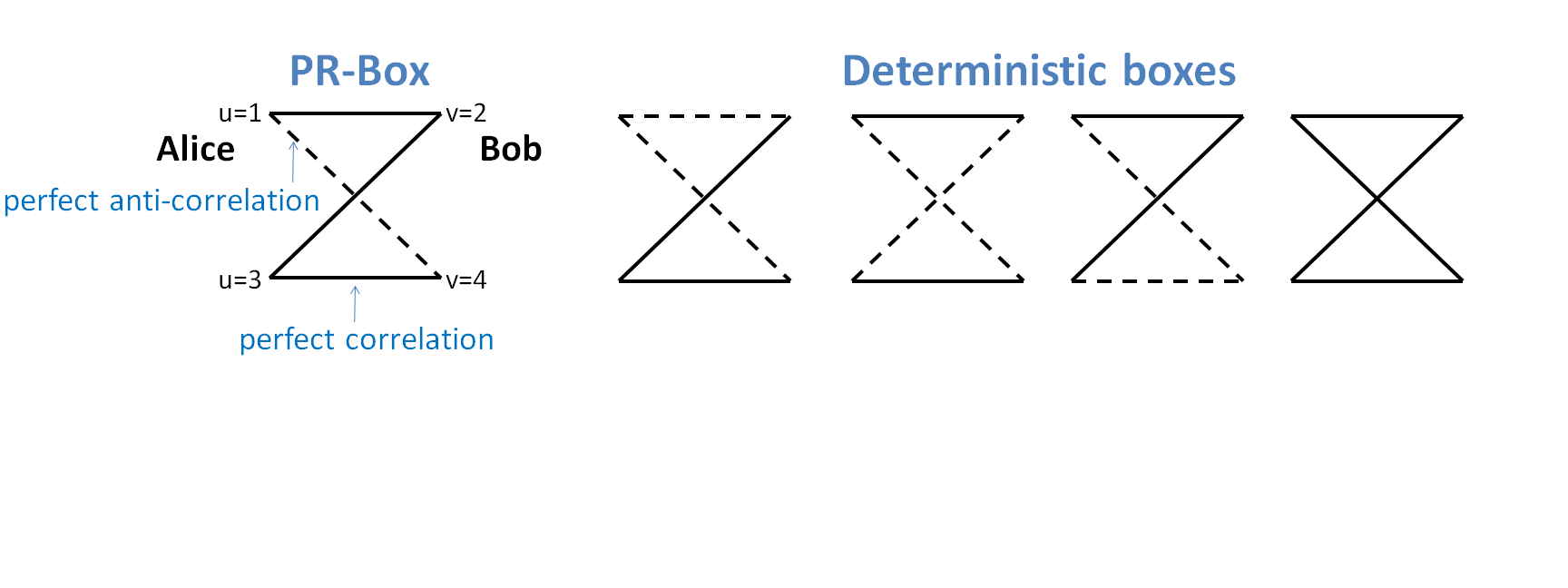}
		\caption{%
			Examples of bipartite boxes with binary inputs and outputs denoted by graphs.
			The Popescu-Rohrlich box (on the left) and local (deterministic) boxes (on the right). {The graphs should be read in the following way: solid (dashed) lines between arbitrary vertices $u$ and $v$ imply that, given the input $I=(u,v)$, the output bits are perfectly correlated (anti-correlated) with probability $1$. This means that e.g. the Popescu-Rohrlich box presented in this figure is determined by the following correlations: $P(O=(0,0)|I=(1,2))=P(O=(1,1)|I=(1,2))=0.5,\; P(O=(0,0)|I=(3,2))=P(O=(1,1)|I=(3,2))=0.5,\;P(O=(0,0)|I=(3,4))=P(O=(1,1)|I=(3,4))=0.5,\;P(O=(0,1)|I=(1,4))=P(O=(1,0)|I=(1,4))=0.5$.}
		}
		\label{fig:pr_box_}
	\end{figure}

\begin{figure*}[!t]
	\centering
	\includegraphics[trim=0cm 0.5cm 0cm 0cm, width=10cm]{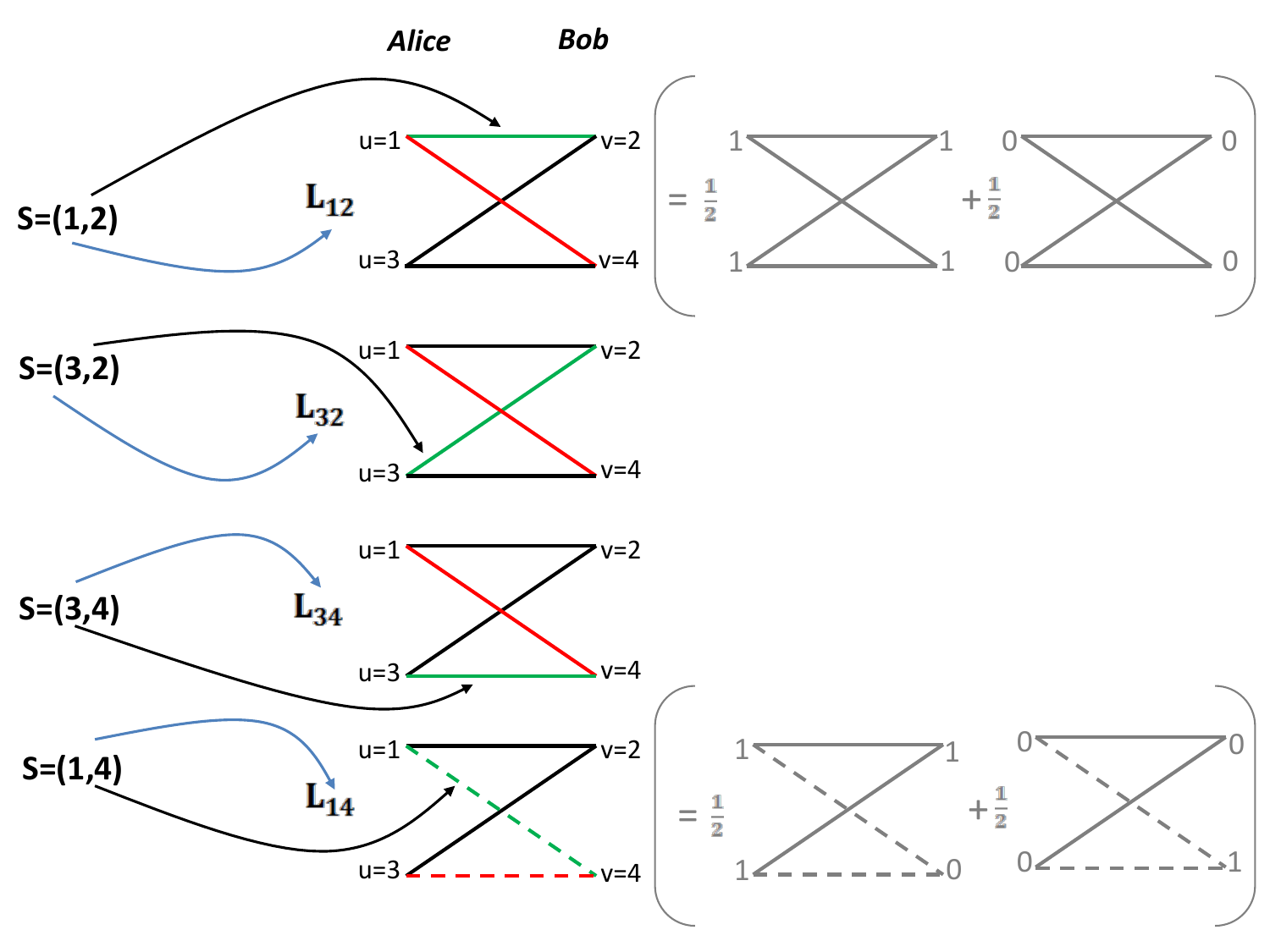}
	\caption{
		Bits from an $\varepsilon$-SV source (on the left) are perfectly correlated with local boxes supplied to honest parties (on the right). Correlations described by Eq. (\ref{correal:box}) are indicated by blue arrows. Additionally, bits from an $\varepsilon$-SV source are perfectly correlated with the inputs to boxes (see Eq. (\ref{correl:inputs})), which is indicated by black arrows. These correlations allow only for measuring green edges and hence Alice and Bob always observe an optimal Bell value. If red edges could be measured, the locality of boxes would be detected.
	}
	\label{fig:ex_new}
\end{figure*}

Imagine that Alice and Bob share a box $L$ which is a mixture of local boxes $L_{ij}$, where $i = 1,3$ labels Alice's inputs and $j = 2,4$ labels Bob's inputs: 
\begin{align}
	L=\frac{1}{4}\left(L_{12}+L_{32}+L_{34}+L_{14}\right).
\end{align}
(See Fig. \ref{fig:pr_box_} where the PR box and local deterministic boxes are presented and Fig. \ref{fig:ex_new}, where the boxes $L_{ij}$ are given explicitly). The bits from an $\varepsilon$-SV source are perfectly correlated to local boxes as
\begin{align}
\begin{aligned}
	\label{correal:box}
	P\left(L_{ij}|S=(k,l)\right)=\delta_{ik;jl} 
	=\left\{\begin{array}{ll}1,&i=k \;\&\; j=l,\\ 0,&\text{otherwise,}\end{array}\right.
\end{aligned}
\end{align}
where $S$ is the random variable describing bits from an $\varepsilon$-SV source.

In the protocols proposed so far such as \cite{cr}, \cite{ghhhpr}, it is demanded that $I$ and $S$ are perfectly correlated, i.e.
\begin{align}
\begin{aligned}
	\label{correl:inputs}
	P(I=(u,v)|S=(k,l))=\delta_{uk;vl}
	=\left\{\begin{array}{ll}1,&u=k \;\&\; v=l,\\ 0,&\text{otherwise,}\end{array}\right.
	\end{aligned}
\end{align}
which means that bits the from the $\varepsilon$-SV source are used as inputs to the box. 
All the correlations are indicated in Fig. \ref{fig:ex_new}. Now, we see that although the box $L$ is manifestly local, the honest parties do not detect it in the protocols proposed so far. Indeed, correlations (\ref{correal:box}) and (\ref{correl:inputs}) imply that input $I=(k,l)$ may only be introduced to box $L_{kl}$, adapted exactly to this input, so that $L$ mimics the action of the PR box on any input. On the other hand, if there was independence between the $\varepsilon$-SV source and the boxes, the parties would recognize that the object $L$ is local.

{%
	To conclude, this toy example clearly illustrates that perfect correlation of inputs and devices excludes any possibility of randomness amplification. To circumvent this type of attack, we introduce the SV-condition for boxes, which is the weakest assumption (thus far) that still allows for randomness amplification.
}

{In the next section we show that the SV-condition for boxes implies the following: if the Bell value observed by the honest parties (\( \delta^{\text{obs}} \)) is small, then the true Bell value (\( \delta^{\text{true}} \)) is also small. In Section \ref{sec:exampleTruevsObsChain} we apply the whole reasoning to the chained Bell inequality. Finally, in Section \ref{sec:composableDistance}, we estimate a composable distance between a bit obtained from a single box and a fully random bit. The bound is given as a function of \( \delta^{\text{obs}} \) and $n$  (the number of input pairs considered in the chained Bell inequality)}.

\section{SV-condition for no-signaling boxes and the relation between the true and the observed Bell value}
\label{sec:theObservedBellValuefromSVcondition}

\subsection{Correlations between the source and the device: boxes determined by the source}

Let $S$ denote a random variable which describes an arbitrary portion of subsequent bits from an $\varepsilon$-SV source. Recall that we write $I$ and $O$ for variables which describe the inputs and outputs of the device, respectively. Suppose that bits from an $\varepsilon$-SV source are delivered and (simultaneously) boxes, that are possibly correlated to them, are supplied. Hence, our object of study is
\begin{align}
	\label{def:box}
	P(O|I,S).
\end{align}

\begin{figure*}[!t]
	\centering
	\includegraphics[trim=0cm 0.2cm 0cm 0cm, width=10cm]{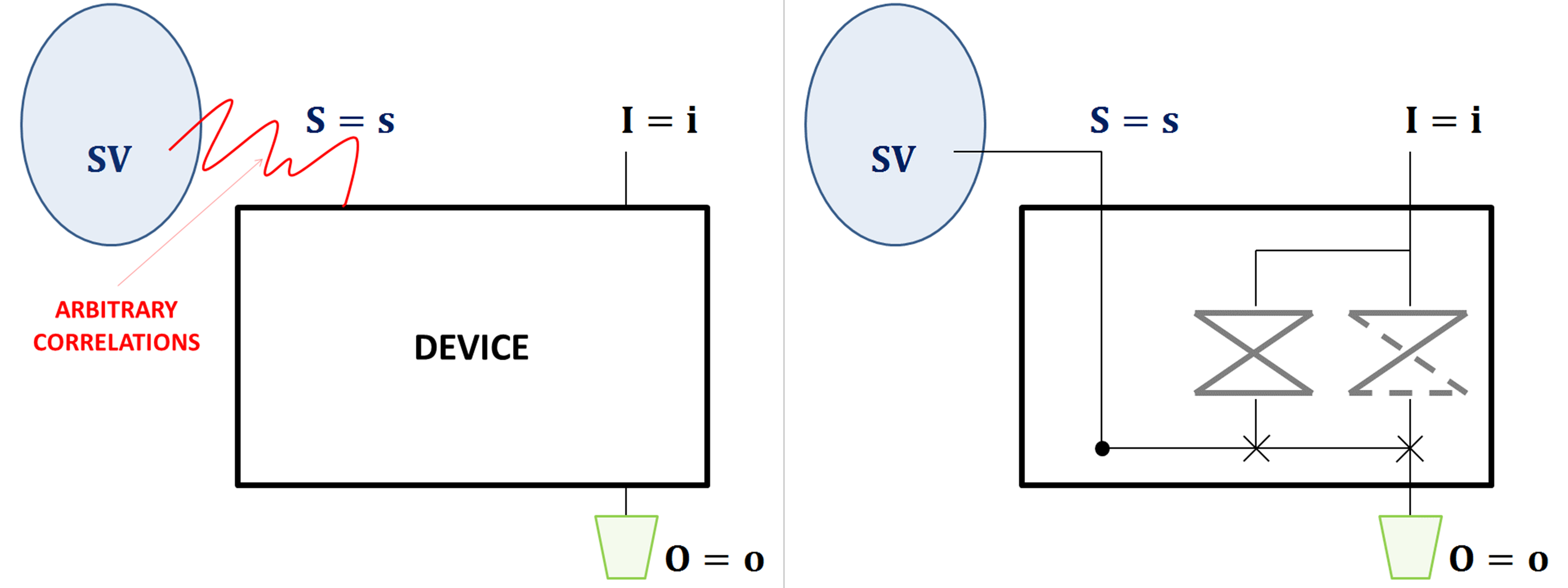}
	\caption{%
		A priori we allow arbitrary correlations between a source and a box (left). To illustrate how malicious these correlations may be, we recall the example described in Section \ref{sec:motivation} (right). Bits from an $\varepsilon$-SV source determine from which box the final output bit is taken. In general, arbitrary input bits may be introduced to the box. The illustration for other Bell inequalities may be more complicated, but the idea is the same.
	}
	\label{fig:scenario_new}
\end{figure*}
Note that $S$ determines how the device acts inside (see Fig. \ref{fig:scenario_new}). 
\begin{remark}
	\label{rem:conditional_vs_joint}
	{%
		Even if conditional distributions of the form $P(O=o|I=i,S=s)$ are equal for arbitrary $o,i,s$, joint distributions $P(O=o,I=i,S=s)$ do not have to be the same.
	}
	{%
		This is just a fact which follows from conventional and meaningful way of thinking about any devices.
	}
\end{remark}



\subsection{SV-condition for boxes}

Let us now precisely state the main assumption used in this paper, which we call the SV-condition for boxes.
{%
} 
{Let $S'$ be a variable describing a portion of bits (disjoined from $S$)  chosen from the same $\varepsilon$-SV source, from which the input $I$ to the device is taken}. Note that we do not assume any temporal ordering between $S$ and $S'$. Let $\eta_{\min}, \eta_{\max} \in (0,1)$ be some functions of $\varepsilon>0$ and $|\mathcal{I}|$ (denoting the number of measurement settings). 
{%
	Although we a priori allow for arbitrary correlations between the~source and the device, 
	there is one constraint which we impose, namely that if $S'=s'$ is input into the device with $\eta_{\min} \leq P(S=s|S'=s') \leq \eta_{\max}$, then $S$ cannot be guessed perfectly even after knowing the output $O=o$, i.e., for every realizations $o,s,s'$ 
	\begin{align}
	\begin{aligned}
		\label{eq:SVforBoxes_S}
		&\eta_{\min}\leq P(S=s|O=o,S'=s')\leq \eta_{\max} \; \text{ for }\;S, S'\\ 
		&\text{such that }\; \eta_{\min}\leq P(S=s|S'=s')\leq 	\eta_{\max}.
	\end{aligned}
	\end{align}
}

\begin{remark}
{It should be noted that any conditional probability is well defined only if the event in its condition is of non-zero probability. Therefore we can consider $P(S=s|O=o, S'=s')$ only for $o$ and $s'$ such that $P(O=o,S'=s')\neq 0$, which means that for an input $s'$ we can obtain an output $o$ with some positive probability.}
\end{remark}

\begin{remark}
	\label{rem:neglecting_e}
	{%
		The distribution remains unchanged even if conditioned upon a variable $e$, which represents some information prior to $S'$. To avoid unnecessary notation, we neglect it in the condition, since it is irrelevant in what follows.
	}
\end{remark}

Assuming condition (\ref{eq:SVforBoxes_S}), which we henceforth call the SV-condition for boxes, we certainly assume less than independence between the source and the device.
Note that the SV-condition for boxes is clearly violated in the toy example from Section \ref{sec:motivation}. Indeed, suppose that there are some testers who obtain further bits from the SV source denoted by the variable $S'$ (so that $p_{\min} \leq P(S'=s'|S=s) \leq p_{\max}$ and conversely $\zeta_{\min} \leq P(S=s|S'=s') \leq \zeta_{\max}$ for some $\zeta_{\min}, \zeta_{\max} \in (0,1)$, 
whose explicit forms are derived in Appendix \ref{appendix1}) and input them into the box. When they input $S'=s$ and observe an output that does not mimic the PR box, which we denote by $O \neq o_{PR}$, then due to the perfect correlations between $S$ and $L_{ij}$ they know that $S \neq s$, i.e., we have
\begin{align}
	P(S=s|S'=s,O\neq o_{PR})=0,
\end{align}
which violates Eq. (\ref{eq:SVforBoxes_S}). 
{Finally, note that taking $S'$ subsequently to $S$ is just the worst case scenario (since $\zeta_{\min}\leq p_{\min}\leq p_{\max}\leq \zeta_{\max}$).}

\subsection{Comparing our assumptions with the assumptions of Chung, Shi and Wu}

{Let us now describe how the SV-condition for boxes, assumed in this paper, differs from what has been assumed in other papers so far}. Firstly, note that to retain the possibility of randomness amplification, one has to necessarily make some assumptions on the correlations between the source and the device (cf. the attack in Section \ref{sec:motivation}). 
The intuition behind the possible assumptions is the following: 
no one in the world should get to know the value of the bits from the SV source better than up to $\varepsilon$ (of course without revealing the bits themselves), even if conditioned upon any possible event in the universe. 
{In particular, if we input a portion of bits from the SV source into any available device and record the outputs, then still any other portion of bits should obey the SV source condition.} 

A stronger assumption that one may consider, is that for an input to the device that is \textit{independent} of the SV source, when conditioned on the output, the source should remain an SV source (see Fig. \ref{fig:scenario_SV2} on the left). This condition is analogous to a~similar condition on min-entropy sources, which is derived from the assumption by Chung, Shi and Wu (CSW) in \cite{Chung-Shi-Wu}. Namely, CSW consider a quantum scenario, where the device $D$ and the  min-entropy source $S$ are correlated as in the cq-state $\rho_{SD}$,
\begin{equation}
	\rho_{SD} := \sum_s P(S=s) |s \rangle \langle s| \otimes \rho^{D}_s
\end{equation}
and they assume that the quantum conditional min-entropy $H_{\text{min}}(S|D)_{\rho}$ of the source conditioned on the device is greater than some constant $k$. This implies (see \cite{KRS}) that for any POVM measurement $\{\mathcal{M}_s\}$ performed by an agent on the quantum register $D$, the probability of the agent correctly guessing $S$, $P_{\text{guess}}(S|D)$, is upper bounded.  
The assumption of Chung, Shi and Wu thus implies that for any input variable $I$ independent of the source $S$, the probability $P_{\text{guess}}(S|D)$ obeys
\begin{align}\label{eq:csw}
\begin{aligned}
P_{\text{guess}}(S|D)= \sum_s P(O=s|I=i) P(S=s|O=s, I=i)\leq 2^{-k} 
\end{aligned}
\end{align}
for all $i$. 
Correlations between the source and the device are also limited by the condition similar to (\ref{eq:csw}) in the more recent paper by Chung, Shi and Wu \cite{chungNew}.

{Condition (\ref{eq:csw})} (whether in the scenario of a min-entropy source, or that of an SV source) has the drawback of effectively introducing an agent that is not correlated with the weak source. However we know that from two independent partially random sources one can extract perfect randomness in the classical world. So the operational realization of the originally mathematical condition might require the existence of an independent variable, implying the possibility of obtaining randomness right from the source and the agent's variable, if the latter's distribution was not deterministic.

{
The most orthodox approach, which is free from the above drawback is the following: since no-one in the world can choose a measurement of his/her own free will, the only way to choose it is to use some weak source. This concept is used to weaken the original assumption (saying that, conditioned on any measurement, an SV source stays the same). Namely, we imagine that an agent draws bits from the SV source and chooses measurements according to these bits (see Fig. \ref{fig:scenario_SV2} on the right).  
The new condition is clearly weaker than the original one, because it can be reduced to it (by assuming that an SV source should stay the same for any joint distribution of choice of measurements and bits from the source).}

In this paper, we consider a somewhat intermediate scenario {(see Fig. \ref{fig:scenario_SV2} in the middle)}: we assume that the agent (which we call the "tester") has a~variable which describes subsequent bits drawn from the same SV source (so that his variable will not be necessarily independent of the other portion of the SV source, used as input by the users who want to draw randomness).
However, we also assume that  the device is correlated with the tester's variable only through the users' variable, i.e. that for any $o,i,s,s'$ we have $P(O=o|I=i, S = s,S'=s')= P(O=o|I=i, S=s)$. 

\begin{figure*}[!t]
	\centering
	\includegraphics[trim=0cm 0.2cm 0cm 0cm, width=15cm]{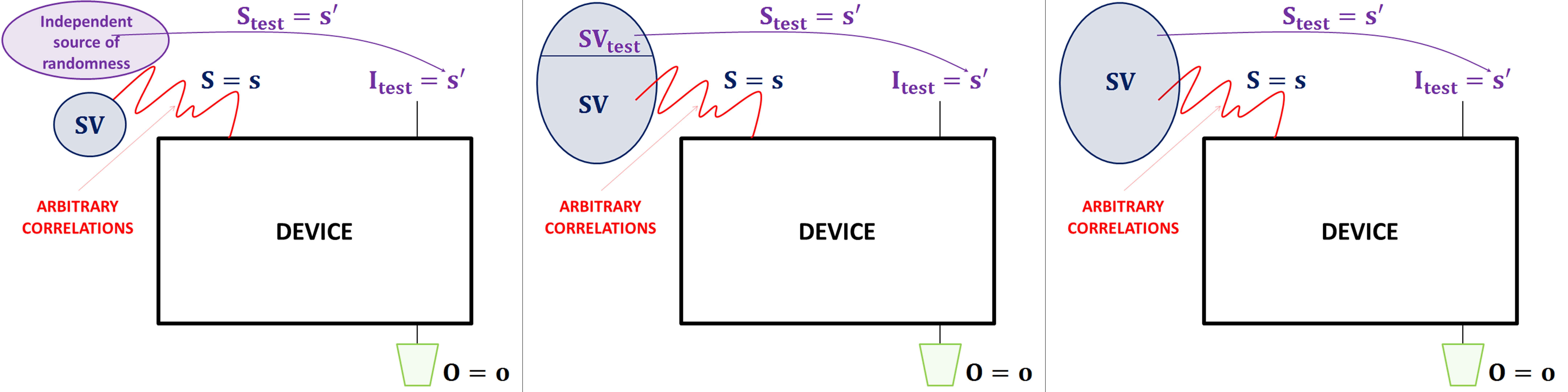}
	\caption{%
		{(Left) The $\varepsilon$-SV source represented by the variable $S$ is correlated to the device, so that $S$ determines the box. The SV-condition for boxes in Eq.(\ref{eq:SVforBoxes_S}) is verified using the additional source of randomness, which is independent of the given $\varepsilon$-SV source. The figure corresponds to the SV-analogue of the CSW condition.
		(Middle) The main part of the $\varepsilon$-SV source represented by the variable $S$ is correlated to the device. Other bits denoted by the variable $S_{\text{test}}$ from the part of the $\varepsilon$-SV source SV$_{\text{test}}$ are correlated with the device only through the variable $S$. If bits are taken from SV$_{\text{test}}$ and used as inputs to the device, one can check whether the SV-condition for boxes is  violated or not. 
		(Right) The desired scenario, in which the SV-condition for boxes is verified using the same weak source of randomness, which is used by the honest parties in the protocol (probably a more sophisticated testing procedure is required here.}
	}
	\label{fig:scenario_SV2}
\end{figure*}

This is a~clearly weaker assumption than the SV-analogue of the CSW condition, since if we take $S'$ to be independent of $S$, we obtain {the SV-analogue of the CSW condition}, while in our case this condition need not be met, and the dependence between $S'$ and $S$ may be chosen by an adversary. In other words, in the SV analogue of the CSW assumption, one requires that for some particular joint distribution (with independent $I$ and $S$), $P(S|I, O)$ is still an SV source, irrespective of the protocol, while  in our case, the latter may hold for some other distribution, this time chosen \textit{adversarially} for any given protocol. 
 
{It should be noted that our reasoning, based on weakening the independence assumption, cannot be applied to arbitrary min-entropy sources, since after gaining knowledge about some bits from a general H$_{\min}$ source, the rest of the bits need not constitute an H$_{\min}$ source any more. On the other hand, the CSW proof could still apply to block min-entropy sources (see \cite{block_min}), i.e., sources that are divided into blocks such that each block has a min-entropy at least $k > 0$, conditioned upon the value of the other blocks. The investigation of the class of min-entropy sources for which the weaker condition still allows for the possibility of randomness amplification and the applicability of the CSW proof under this condition are left as open questions.} 
 
The threshold for the range of $\varepsilon$ for which we will be able to amplify the SV source in the present paper (obtained in Theorems \ref{thm:epsilon1} and \ref{thm:epsilon2}) is weaker than the one obtained by Colbeck and Renner in \cite{cr}. 
This however is only to be expected as the scenario considered in this paper is more general than the scenario analyzed in \cite{cr}, which was based on the assumption that the source and the device are independent. While the protocols of \cite{acin}, \cite{brghhh} and \cite{ravi} achieve randomness amplification for the entire range of $\varepsilon$ and the latter two protocols also tolerate noise within a finite-device framework, they also do so under the assumption of independence between source and device and are therefore incomparable with the results in this paper.

\subsection{Scenario}

The scenario is as follows. There are: an $\varepsilon$-SV source and a device correlated to some portion of subsequent bits from the source, described by the variable $S$ (see Fig. \ref{fig:scenario_honest}). 
The honest parties draw $S=s$ from the source and use it as an input to the box, which means that $S$ and $I_{HP}$, the random variable describing the measurement settings of the honest parties, are perfectly correlated, i.e.
\begin{align}
	\label{eq:P(I|S)}
	P(I_{HP}=i|S=s)=\delta_{is}\qquad \text{for every }i,s.
\end{align} 
The honest parties then test the statistics of a box for suitable violation of a certain Bell inequality. 

	\begin{figure}[!t]
		\centering
		\includegraphics[trim=0cm 0.2cm 0cm 0cm, width=5cm]{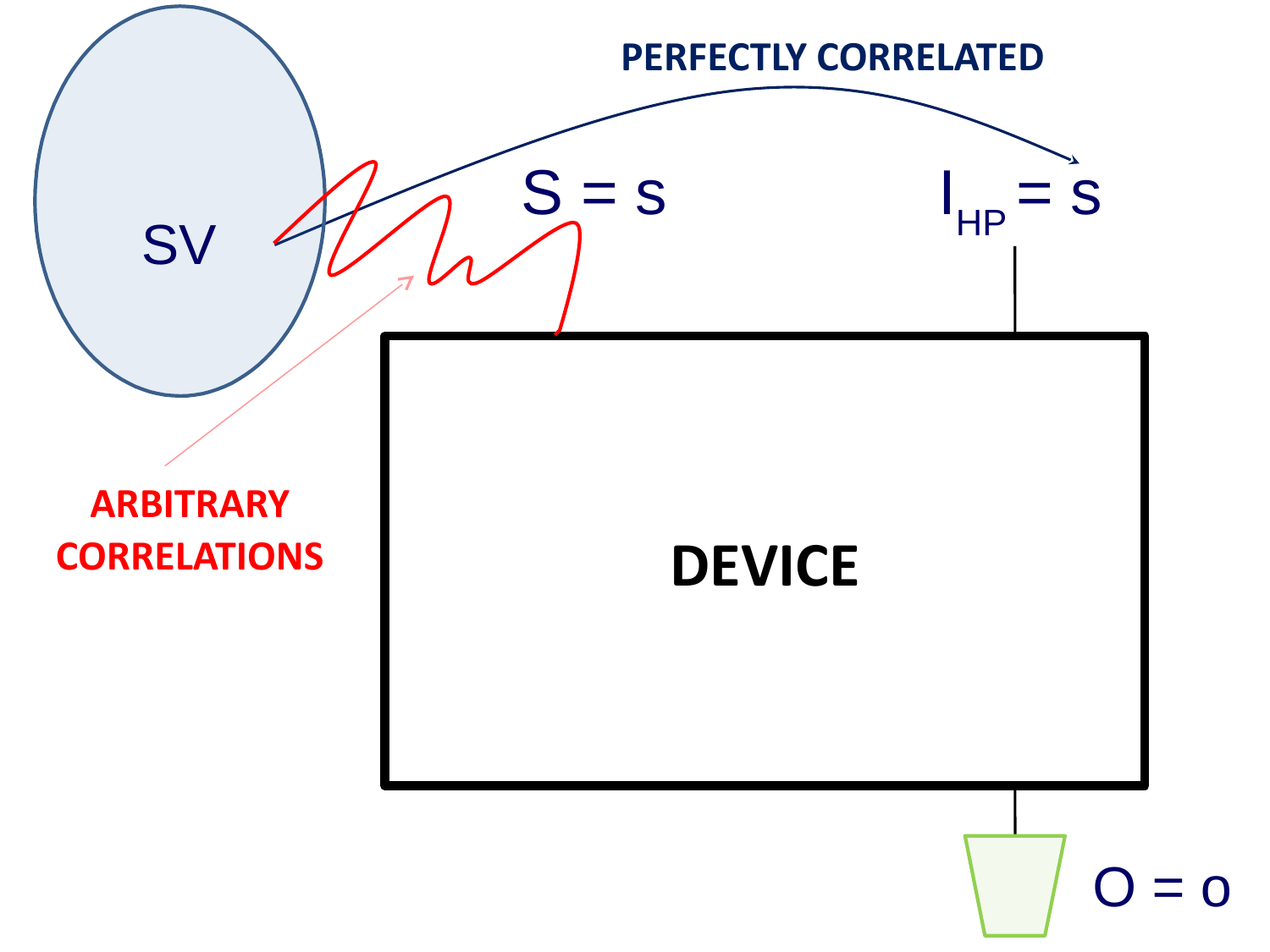}
		\caption{%
			Bits from an $\varepsilon$-SV source are used by honest parties as inputs. The correlation is given by Eq.~(\ref{eq:P(I|S)}).
		}
		\label{fig:scenario_honest}
	\end{figure}

\subsection{The true and the observed Bell value}
\label{sseq:old:3e}

In the most general form, the Bell value is given by the formula
\begin{align}
	\label{def:Bell_value_general}
	\delta=\sum_{o,i}P(O=o,I=i)B(i,o),
\end{align}
{%
	where $B$ is an indicator vector for the Bell inequality and $P$ is an arbitrary joint probability distribution. We specify it depending on the context.
}

We are particularly interested in evaluating the true Bell value, as it informs us whether the box delivers randomness or not. Let $\mathcal{I}$ denote all the settings appearing in the Bell expression. The true Bell value $\delta^{\text{true}}$ is calculated for variables $I_{\text{indep}}$, uniformly distributed ($P(I_{\text{indep}}=i)=1/|\mathcal{I}|$) and independent from $S$. It is then defined as follows:
\begin{align}
	\label{def:Bell_value_true}
	\delta^{\text{true}}
	=\frac{1}{|\mathcal{I}|}\sum_{o,i}P(O=o|I_{\text{indep}}=i)B(i,o),
\end{align}
where $|\mathcal{I}|$ is the number of measurement settings. 

Further, we define the observed Bell value, i.e. we write Eq. (\ref{def:Bell_value_general}) for $I_{HP}$, determined by Eq. (\ref{eq:P(I|S)}), and obtain
\begin{align}
	\label{def:Bell_value_obs}
	\delta^{\text{obs}}_{HP}=\sum_{o,s}P(S=s)P(O=o|I_{HP}=s,S=s)B(s,o).
\end{align}

The aim is to show that the true Bell value is small whenever the observed value is small, i.e. the ratio $\delta^{\text{obs}}_{HP}/ \delta_{SV}^{\text{true}}$ is controlled.

\subsection{Testing the SV-condition for boxes}

Honest parties test the statistics of a box using a certain Bell inequality. There is a danger that they may be cheated, as exemplified in Section \ref{sec:motivation}. The $\varepsilon$-SV source can be  correlated with the device, as illustrated in Fig. \ref{fig:scenario_SV2}.

Since the honest parties only input $I_{HP}$ which is perfectly correlated to $S$, 
\begin{align}
	P(I_{HP}=i|S=s)=\delta_{is},
\end{align}
they are themselves not able to verify whether the SV-condition for boxes (\ref{eq:SVforBoxes_S}) is violated or not. Therefore, we consider testers who have access to part of the $\varepsilon$-SV source (SV$_{\text{test}}$), described by the variable $S_{\text{test}}$, which is correlated with the device only through the variable $S$ and does not change the statistics of a box $P(O|I,S)$ (see Fig. \ref{fig:scenario_SV2}, in the middle), i.e.
\begin{align}
	\label{eq:P(I2|S)}
	p_{\min}\leq P(S_{\text{test}}=s'|S=s)\leq p_{\max}\;\text{ for every }s,s'
\end{align}
and 
\begin{align}
	\label{eq:box-subsSV}
	P(O|I,S,S_{\text{test}})=P(O|I,S).
\end{align}

{
\begin{remark}\label{remark:technical}
Note that by assuming that an $\varepsilon$-SV source should stay the same for any joint distribution of choice of measurements and bits from the source, one can simply choose an independent source of randomness as a testing part SV$_{\text{test}}$. Therefore the condition proposed in this paper, although a bit technical, leads us to the desired scenario, in which the device is tested (in terms of satisfying the SV-condition for boxes \ref{eq:SVforBoxes_S}) using the same weak source of randomness, which is used by the honest parties in  the protocol (cf. Fig. \ref{fig:scenario_SV2}).
\end{remark}
}

When honest parties take the portion of bits $S$ from the main part of source (they do not have access to SV$_{\text{test}}$), to which the device is possibly correlated, the testers may be asked to perform the measurement using their bits $S_{\text{test}}$ as input, i.e.
\begin{align}
	\label{eq:test-input}
	P(I_{\text{test}}=i'|S_{\text{test}}=s')=\delta_{i's'}.
\end{align}
The overall picture is now the following. We have two different joint distributions $P(O,I,S,S_{\text{test}})$ and\\
$P(O,I_{\text{test}},S,S_{\text{test}})$. Conditional distributions are correlated as follows:
\begin{align}
\begin{aligned}
	P(O=o|I=i,S=s,S_{\text{test}}=s')
	&\stackrel{\text{Eq.(\ref{eq:box-subsSV})}}{=}P(O=o|I=i,S=s)\\
	&\overset{\text{Remark \ref{rem:conditional_vs_joint}}}{=}P(O=o|I_{\text{test}}=i,S=s)\\
	&\stackrel{\text{Eq.(\ref{eq:box-subsSV})}}{=}P(O=o|I_{\text{test}}=i,S=s,S_{\text{test}}=s')
\end{aligned}
\end{align}
for every $o,i,i',s,s'$, where the pairs of variables $I,S$ and $I_{\text{test}},S_{\text{test}}$ are each perfectly correlated.
As shown in Appendix \ref{appendix1}, we have that Eq. (\ref{eq:P(I2|S)}) implies
\begin{align}
	\label{eq:P(S|I2}
	\zeta_{\min}\leq P(S=s|I_{\text{test}}=s') \leq \zeta_{\max},
\end{align}
where $\zeta_{\min}$ and $\zeta_{\max}$ are functions of $p_{\min}$, $p_{\max}$ and $|\mathcal{I}|$, explicitly given by Eq. (\ref{def:zeta_min_max}) in Appendix \ref{appendix1}. Due to the SV-condition for boxes (\ref{eq:SVforBoxes_S}) this gives that
\begin{align}
	\label{eq:SV_Itest}
	\zeta_{\min}\leq P(S=s|I_{\text{test}}=s',O=o) \leq \zeta_{\max}
\end{align}
for every $s$, $s'$, $o$.

	

We now introduce an intermediate value between $\delta^{\text{obs}}_{HP}$ and $\delta^{\text{true}}$:
\begin{align}
	\label{def:Bell_value_SV}
	\delta^{\text{true}}_{SV}=\sum_{o,s'}P(O=o,I_{\text{test}}=s')B(s',o),
\end{align}
where $I_{\text{test}}$ is a random variable satisfying Eq. (\ref{eq:P(I2|S)}). 
Note that, 
{%
	according to the observation in Remark \ref{rem:conditional_vs_joint}
}
, we obtain
\begin{align}
	\begin{aligned}
		\delta^{\text{true}}
		&\overset{\text{Eq. }(\ref{def:Bell_value_true})}{=}
		\frac{1}{|\mathcal{I}|}\sum_{o,i,s}P(S=s)P(O=o|{I_{\text{indep}}=i},S=s)B(i,o)\\
		&\stackrel{\text{Remark} \ref{rem:conditional_vs_joint}}{=}\frac{1}{|\mathcal{I}|}\sum_{o,i,s}\frac{P(S=s)P(O=o,I_{\text{test}}=i,S=s)B(i,o)}{P(I_{\text{test}}=i,S=s)}\\
		&=\frac{1}{|\mathcal{I}|}\sum_{o,i,s}\frac{P(O=o,I_{\text{test}}=i,S=s)B(i,o)}{P(I_{\text{test}}=i|S=s)}
	\end{aligned}
\end{align}
and hence, according to Eq. (\ref{eq:P(I2|S)}) and the definition of $\delta^{\text{true}}_{SV}$ in Eq.(\ref{def:Bell_value_SV}), we have
\begin{align}
	\label{eq:true_vs_trueSV}
	{\frac{1}{p_{\max}|\mathcal{I}|}\delta_{SV}^{\text{true}}\leq \delta^{\text{true}}\leq \frac{1}{p_{\min}|\mathcal{I}|}\delta_{SV}^{\text{true}}.}
\end{align}

\subsection{Results and proofs}

At this point, let us explicitly restate all the assumptions used in the paper for clarity:
\begin{enumerate}
	\item There are spatially separated honest parties who share a no-signaling box, i.e., one constrained by conditions Eq.(\ref{eq:ns}).
	\item Correlations between the source and the device are only limited by the SV-condition for boxes (see Eq. (\ref{eq:SVforBoxes_S})). The device is correlated to the main part of the source from which honest parties draw their bits represented by variable $S$ (see Eq. (\ref{eq:P(I|S)})). 
	\item There exists another part of the source, called SV$_{\text{test}}$, which may only be used (by testers) to verify whether the SV-condition for boxes is violated. $S_{\text{test}}$ drawn from SV$_{\text{test}}$ is only correlated with the device through the variable $S$ and does not change the statistics of the box as given in Eq.(\ref{eq:box-subsSV}) {(cf. Fig. \ref{fig:scenario_SV2} and Remark \ref{remark:technical})}. 
\end{enumerate}
The main result of this Section is the following.
\begin{theorem}
	\label{thm:ratio}
	Under assumptions 1-3 we obtain
	\begin{align}
		\label{eq:ratio}
		\frac{\delta^{\text{obs}}_{HP}}{\delta^{\text{true}}}\geq |\mathcal{I}|\frac{p_{\min}\zeta_{\min}}{p_{\max}}.
	\end{align}
\end{theorem}

\begin{proof}
	Note that Eqs. (\ref{eq:SV_Itest}) and (\ref{eq:P(I2|S)}), as well as Remark \ref{rem:conditional_vs_joint}, imply that
	\begin{align}
		\begin{aligned}
			\delta^{\text{obs}}_{HP}
			&\overset{\text{Eq. }(\ref{def:Bell_value_obs})}{=}
			\sum_{o,s}P(S=s)P(O=o|{I_{HP}=s},S=s)B(s,o)\\
			&\stackrel{\text{Remark} \ref{rem:conditional_vs_joint}}{=}\sum_{o,s}P(S=s)P(O=o|{I_{\text{test}}=s},S=s)B(s,o)\\
			&=\sum_{o,s}P(S=s)\frac{P(O=o,S=s|I_{\text{test}}=s)}{P(S=s|I_{\text{test}}=s)}B(s,o)\\
			&=\sum_{o,s}\frac{P(S=s)P(I_{\text{test}}=s)}{P(S=s,I_{\text{test}}=s)}P(S=s|O=o,I_{\text{test}}=s)\\
			&\qquad P(O=o|I_{\text{test}}=s)B(s,o)\\
			&\overset{\text{Eq. }(\ref{eq:SV_Itest})}{\geq} 
			\zeta_{\min}\sum_{o,s}\frac{P(O=o,{I}_{\text{test}}=s)}{P({I}_{\text{test}}=s|S=s)}B(s,o)\\
			&\overset{\text{Eq. }(\ref{eq:P(I2|S)}), \text{Eq.}(\ref{eq:test-input})}{\geq} 
			\frac{\zeta_{\min}}{p_{\max}}\sum_{s,o}P({O}=o,{I}_{\text{test}}=s)B(s,o)
			\overset{\text{Eq. }(\ref{def:Bell_value_SV})}{=}
			\frac{\zeta_{\min}}{p_{\max}}\delta^{\text{true}}_{SV}.
		\end{aligned}
	\end{align}
	Referring to Eq. (\ref{eq:true_vs_trueSV}), we obtain
	\begin{align}
		\begin{aligned}
			\delta^{\text{obs}}_{HP}
			\geq |\mathcal{I}|\frac{p_{\min}\zeta_{\min}}{p_{\max}}\delta^{\text{true}},
		\end{aligned}
	\end{align}
	which completes the proof.
\end{proof}

\begin{remark}
	\label{th:single_box_general}
	Suppose that assumptions 1-3 are satisfied. 
	Note that any Bell value (of non-local boxes) observed in a lab can be predicted by the rules of quantum mechanics and hence we set
	\begin{align}
		\delta_{HP}^{\text{obs}}=\delta_Q.
	\end{align}
	Further, due to Theorem \ref{thm:ratio}, we obtain
	\begin{align}
		\label{eq:single_box_crucial}
		\delta^{\text{true}}\leq \delta_Q\frac{ p_{\max}}{|\mathcal{I}|p_{\min}\zeta_{\min}},
	\end{align}
	where $\zeta_{\min}$, $p_{min}$ and $p_{\max }$ depend on both $|\mathcal{I}|$ and $\varepsilon$. 	
	The above inequality allows to set an upper bound for $\varepsilon$ (as $|\mathcal{I}|\to \infty$), as illustrated in the example of the chained Bell inequality below.
\end{remark}

\section{Example - the true versus the observed Bell value of the chained inequality}\label{sec:chain}
\label{sec:exampleTruevsObsChain}

\subsection{The chained Bell inequality}\label{subsec:chain}

The chained Bell inequality considers the bipartite scenario of two spatially separated parties Alice and Bob. 
Let $n\in \mathbb{Z}_{+}$ be an arbitrary positive even integer. 
Let the sets $U_A := \{1,3,\ldots,{n-1}\}$ and $U_B := \{2,4,\ldots,n\}$ correspond to the measurement settings chosen by Alice and Bob, respectively. 
Each measurement pair $(u,v)$, where $u\in U_A$, $v\in U_B$, results in a~binary outcome $x\in\{0,1\}$ for Alice, and $y\in\{0,1\}$ for Bob. The chained Bell inequality is then written as \cite{braunstein_caves}
\begin{align}
	\label{bell}
	\begin{aligned}
		\frac{1}{n}\Bigg(\sum_{u,v:|u-v|=1}P(O=(x,y)|I=(u,v))[x\oplus y=1]
		+P(O=(x,y)|I=(1,n)) [x\oplus y=0]\Bigg)
		\geq \frac{1}{n},
	\end{aligned}
\end{align}
where $\oplus$ denotes addition modulo $2$ and $[B]$ denotes the Iverson bracket taking value $1$ when $B$ is true and $0$ otherwise.

\begin{remark}
	\label{edges_in_chain}
	Note that out of the $n^2/4$ possible measurement pairs, only $n$ neighbouring pairs, forming a~chain, are considered in the inequality. 
\end{remark}

For clarity, we further label the pairs of inputs by the number of the edge in the chain (see Remark $\ref{edges_in_chain}$), i.e., instead of a pair $(u,v)$, where $u\in U_A$, $v\in U_B$ and $|u-v|=1$, we set $i := \min\{u,v\}$. Similarly, the remaining pair in a chain $(1,n)$ is denoted by $n$. 
Note that the true Bell value for an~arbitrary box~$P$ is then given by 
\begin{align}
	\begin{aligned}
		\delta^{\text{true}}(P)
		=\frac{1}{n}\Bigg(\sum_{i\neq n}P(O=(x,y)|I=i)[x\oplus y=1]
		+P(O=(x,y)|I=n)[x\oplus y=0]\Bigg),
	\end{aligned}
\end{align}
while the observed value is of the form
\begin{align}
\begin{aligned}
	\delta^{\text{obs}}_{AB}(P)
	&=\sum_{s\neq n}P(S=s)P(O=(x,y)|I=s,S=s)[x\oplus y=1]\\
	&\quad\;+P(S=n)P(O=(x,y)|I=n,S=n)[x\oplus y=0].
	\end{aligned}
\end{align}

We recall that results observed in a lab are not better than the values predicted by the rules of quantum mechanics. Quantum mechanics violates $(\ref{bell})$ and provides a value of 
\begin{align}\label{def:delta_Q}
	\begin{aligned}
		\delta_Q := \sin^2(\pi/2n),
	\end{aligned}
\end{align}
which tends to $0$, as $n\to \infty$, with a rate of convergence $1/n^2$. 
This optimal quantum value is obtained by measuring on the maximally entangled state $|\phi^+\rangle=\frac{1}{\sqrt{2}}(|00\rangle+ |11\rangle)$ with the measurement settings defined by the bases $\{|\alpha\rangle,|\alpha+\pi\rangle\}$, $\alpha\in\frac{\pi}{n}\{0,2,\ldots,n-2\}$, for Alice and $\{|\beta\rangle,|\beta+\pi\rangle\}$, $\beta\in\frac{\pi}{n}\{1,3,\ldots,n-1\}$, for Bob, where $|\cdot\rangle=\cos(\cdot/2)|0\rangle+\sin(\cdot/2)|1\rangle$. 

\subsection{Value of chained Bell inequalities on boxes}

While testing the chained Bell inequality, we do not distinguish between boxes with the same probability distributions for neighboring pairs of settings. Hence, we consider only two types of extremal boxes: ideal or "bad". Any other box may be represented as a~mixture of these boxes, due to the characterization of the extremal boxes for this scenario in \cite{jm}. 

We call boxes ideal ($P_{\text{ideal}}$) if they violate the chained Bell inequality (\ref{bell}) maximally and give perfectly random bits (boxes $P_{\text{ideal}}$ play for the chained Bell inequality the same role as PR-boxes play for the CHSH inequality). With respect to the probability distributions significant for the chained Bell expression, there is exactly one box violating $(\ref{bell})$ to $0$ (compare with Remark \ref{edges_in_chain}). Precisely, this is the no-signaling box with structure of perfect correlations for the $n-1$ neighboring pairs in the sum and a perfect anti-correlation for the remaining pair $n$ (see \cite{jm} for details).
Then,
\begin{align}
	\label{eq:P_ideal}
	\begin{aligned}
		\delta^{\text{true}}(P_{\text{ideal}})=&\frac{1}{n}\Bigg(\sum_{i\neq n}P_{\text{ideal}}(O=(x,y)| I=i)[x\oplus y=1]\\
		&+P_{\text{ideal}}(O=(x,y)|I=n)[x\oplus y=0]\Bigg)=0.
	\end{aligned}
\end{align}

In classical theory, there are no ideal boxes. The notion $P_{\text{bad}}$ is used for these extremal (local deterministic) boxes whose Bell value is at least $1/n$, which means that there is at least one contradiction with probability distributions of ideal boxes (for neighboring pairs of settings). 
Apart from purely classical boxes there are also other bad boxes which do not violate the chained Bell inequality (\ref{bell}) (some of them even give randomness, but are inappropriate for the chosen inequality (\ref{bell})). 
Convex combinations of boxes $P_{\text{bad}}$ are denoted by $P_{\text{BAD}}$. By convexity,
\begin{align}
	\label{eq:P_BAD}
	\delta^{\text{true}}(P_{\text{BAD}})\geq 1/n.
\end{align}

\begin{remark}
	\label{rem:Lambda_P_ideal_P_bad}
	Any box $P$ is a mixture of boxes which attain an optimal Bell value $0$ and boxes which do not violate the chained Bell inequality 
	\begin{align}
		\label{eq:P=Lambda_P_bad}
		\begin{aligned}
			P=(1-\Lambda_P)P_{\text{ideal}}+\Lambda_P P_{\text{BAD}},\quad \Lambda_P\in[0,1].
		\end{aligned}
	\end{align} 
\end{remark}

\begin{corollary}
	\label{corol:delta_true_Lambda/n}
	The true Bell value for an arbitrary box $P$ is estimated as follows:
	\begin{align}
		\label{eq:delta_true>Lambda/n}
		\delta^{\text{true}}(P)\geq \Lambda_P / n,
	\end{align}
	where $\Lambda_P$ is defined by Eq. (\ref{eq:P=Lambda_P_bad}).
\end{corollary}
\begin{proof}
	Note that, according to Remark \ref{rem:Lambda_P_ideal_P_bad}, we obtain
	\begin{align}
	\begin{aligned}
		\delta^{\text{true}}(P)
		&\overset{\text{Eq. }(\ref{eq:P=Lambda_P_bad})}{=}
		\delta^{\text{true}}\left((1-\Lambda_P)P_{\text{ideal}}+\Lambda_P P_{\text{BAD}}\right)\\
		&=(1-\Lambda_P)\delta^{\text{true}}\left(P_{\text{ideal}}\right)
		+\Lambda_P \delta^{\text{true}}\left(P_{\text{BAD}}\right)\\
		&\overset{\text{Eq. }(\ref{eq:P_ideal})}{=}
		\Lambda_P \delta^{\text{true}}\left(P_{\text{BAD}}\right)
		\overset{\text{Eq. }(\ref{eq:P_BAD})}{\geq} 
		\Lambda_P / n.
	\end{aligned}
	\end{align}
\end{proof}

At this point we explicitly state values of \( p_{\text{min}} \text{ , } p_{\text{max}} \text{ and } \zeta_{\min} \) to be
\begin{align}
\label{def:p_min/max}
\begin{aligned}
p_{\text{min}} := \frac{p_-^{2r}}{np_+^{2r}},
\;
p_{\text{max}} := \frac{p_+^{2r}}{p_+^{2r}+(n-1)p_-^{2r}},
\;
\zeta_{\min}=\frac{p_{\min}^2}{np_{\max}^2}
\end{aligned}
\end{align}
for $r = \log(n/2)$.
The estimates come from \cite{cr} and Appendix \ref{appendix1}. Even more accurate estimates are given in \cite{ghhhpr}.

\section{Composable distance in terms of the chained Bell value}
\label{sec:composableDistance}

Let $I=i$, $i\in\{1,\ldots,n\}$, be any chosen input to a box $P$. 
To measure the distance between 
{an output bit of the box $P$ (for Alice) and a fully random bit}, we introduce the following quantity:
{%
	\begin{align}
		\label{def:d}
		d(P)=\max_i\{d_i(P)\},
	\end{align}
}
where
\begin{align}
	\label{def:di}
	d_i(P)=\left|p_i^{(P)}(0)-1/2\right|+\left|p_i^{(P)}(1)-1/2\right|,
\end{align}
and
{%
	\begin{align}\label{def:d_additional}
		\begin{aligned}
			&p_i^{(P)}(x)=\sum_{y\in\{0,1\}}P\left(O=(x,y)|I=i\right) [x\oplus y=0],\\
			&p_n^{(P)}(x)=\sum_{y\in\{0,1\}}P\left(O=(x,y)|I=n\right)[x\oplus y=1]
		\end{aligned}
	\end{align}
	for $i\in\{1,\ldots,n-1\}$, $x\in \{0,1\}$. 
}
Note that for boxes $P_{\text{ideal}}$, which generate randomness, we have $p_i^{P_{\text{ideal}}}(0)=p_i^{P_{\text{ideal}}}(1)=1/2$ for every~$i$. 
Due to Eq. (\ref{eq:P=Lambda_P_bad}), we further obtain
\begin{align}
	\label{eq:p_P}
	p_i^{(P)}(x)=\Lambda_P p_i^{\text{BAD}}(x)+(1-\Lambda_P)\frac{1}{2},
\end{align}
where $p_i^{\text{BAD}}$ is generated by boxes $P_{\text{BAD}}$ and, in the worst case, it is some deterministic function.

\begin{theorem}
	\label{prop:d_vs_delta}
	Let $d(P)$ be defined by Eq. (\ref{def:d}) for every box $P$ of the form (\ref{eq:P=Lambda_P_bad}). Then
	\begin{align}
	\label{eq:d<Lambda<delta}
		d(P) \leq \Lambda_P\leq n\delta^{\text{true}}(P).
	\end{align}
\end{theorem}

\begin{proof}
	Let us bound the distance $d$ from above. Following Eqs. (\ref{def:d}) and (\ref{eq:p_P}), we obtain
	\begin{align}
		d(P)
		=\Lambda_P{\max_i}\{|p_i^{\text{BAD}}(0)-1/2|+|p_i^{\text{BAD}}(1)-1/2|\}
		\leq \Lambda_P.
	\end{align}

	Note that, due to Eq. (\ref{eq:delta_true>Lambda/n}) of Corollary \ref{corol:delta_true_Lambda/n}, we obtain that
	\begin{align}
		d(P) \leq \Lambda_P\leq n\delta^{\text{true}}(P),
	\end{align}
	which completes the proof and indicates that, 
	{whenever the true Bell value is small, for an arbitrary box $P$, the distribution of an output bit obtained from this box is close to uniform}.
\end{proof}

{So far only a bipartite scenario has been discussed. However, in order to prove security (in a composable way), we have to consider a third party, i.e. an eavesdropper Eve with her input \(w\) and output \(z\). We therefore introduce a tripartite box of the form }
\begin{align}
	P(x,y,z|u,v,w).
\end{align}
{The box satisfies the no-signaling constraints between the honest parties and Eve, i.e.
\begin{align}
P(z|u,v,w)=P(z|w),\; P(x,y|u,v,w)=P(x,y|u,v)
\end{align}
and
\begin{align}
\begin{aligned}
&P(x|u,v,z,w) = P(x|u,z,w),\\ &P(y|u,v,z,w) = P(y|v,z,w).
\end{aligned}
\end{align}
It is then easy to see that}
\begin{align}
	P(x,y|u,v) = \sum_z P(z|w) P(x,y|z,u,v,w).
\end{align}
{Following the original protocol of Colbeck and Renner \cite{cr}, the final random bit is just the output of Alice, so it is enough to consider }
\begin{align}
	\sum_z P(z|w) P(x|z,u,w) = P(x|u).
\end{align}
For the given boxes we can calculate their Bell values. {Finally, we obtain}
\begin{align}
	\label{eq:deltaZW}
	\sum_z P(z|w) \delta^{\text{true}}_{|z,w} = \delta^{\text{true}},
\end{align}
{which, together with Theorem \ref{prop:d_vs_delta}, allows us to estimate the composable distance}. 
{
\begin{definition}
According to \cite{brghhh}, the composable distance between fully random bits and bits $x$ which are derived from a box $P$ (using a private source of weak randomness to generate inputs) is given by
\begin{equation}\label{def:d_c}
		d_\mathrm{c}{(P)} = \sum_x \max_w \sum_z P(z|w) \left|P(x|z, w)-\frac{1}{|X|}\right|,
	\end{equation}
	where \(w\) and \(z\) are Eve input and output, respectively. 
\end{definition}
}

\begin{proposition} 
	\label{prop:distance}
	The composable distance between a fully random bit and a bit obtained as an outcome {of a box \(P\), whose Bell value is described in terms of Eq. (\ref{eq:deltaZW}), is not bigger than} \(2n \delta^{\mathrm{true}}(P) \). 
\end{proposition}
\begin{proof}
	Since Alice output \(x\) is binary, 
	{then the composable distance $d_c$, defined in Eq. \ref{def:d_c}, reads as follows}
	\begin{equation}
		\begin{split}	
			d_\mathrm{c}{(P)} &= \sum_x \max_w \sum_z P(z|w) \left|P(x|z, w)-\frac{1}{2}\right|
			\leq 2\max_w \sum_z P(z|w) \sum_x \left|P(x|z, w)-\frac{1}{2}\right|\\
			&\overset{\text{Eqs. (\ref{def:d})-(\ref{def:d_additional})}}{\leq} 2\max_w \sum_z P(z|w) d\left(P_{|z,w}\right)
			\stackrel{\text{Eq.}\ref{eq:d<Lambda<delta}}{\leq} 2n \max_w \sum_z P(z|w) \delta^{\text{true}}_{|z,w}(P)
			\overset{\text{Eq. (\ref{eq:deltaZW})}}{\leq} 2n\delta^{\text{true}}(P).
		\end{split}
	\end{equation}
\end{proof}

\begin{corollary}\label{corol}
{According to Proposition \ref{prop:distance} and Remark \ref{th:single_box_general}, we obtain
\begin{align}\label{eq:dc<deltaQ}
d_\mathrm{c}{(P)}\leq 2\delta_Q\frac{p_{\max}}{p_{\min}\zeta_{\min}},
\end{align}
where $d_\mathrm{c}$ is a composable distance defined in Eq. (\ref{def:d_c}). Further, the bound on $d_\mathrm{c}{(P)}$ is tending to $0$ (meaning that full randomness of an output bit is guaranteed), as $n\to\infty$, for any 
\begin{align}
\varepsilon<\frac{2^{1/12}-1}{2\left(2^{1/12}+1\right)} \approx 0.0144,
\end{align}
which is proven in Appendix \ref{appendix5}.}
\end{corollary}

\section{Randomness amplification protocol based on the chained Bell inequality}
\label{sec:randomnessAmplificationProtocol}

The protocol is given in Figure \ref{protocol}.
\begin{figure*}[!t]
\centering
\begin{tabular}{|p{16.2cm}|} 
	\hline 
	\begin{center}
 	   \textbf{Protocol}
	\end{center}

	\begin{enumerate}
		\item[1.] The honest parties Alice and Bob choose their measurement settings $u_i\in U_A$, $v_i\in U_B$ for each of the runs $i = 1,\dots, M$ where the input sets are of size $|U_A| = |U_B| = n/2$ (see Section \ref{subsec:chain} for the precise definitions of $U_A$, $U_B$). To do so, in each run {any of them uses $r=\log(n/2)$} bits from an $\varepsilon$-SV source. Simultaneously, a sequence of $M$ boxes is supplied. 
		\item[2.]
		They check that the cardinality $|\mathcal{S}|$ of the set $\mathcal{S}$ defined as
		\begin{align}
		\begin{aligned}
			\mathcal{S} := \big\{i\in\{1,\ldots,M\}:|u_i-v_i|=1\; \vee\; (u_i,v_i)=(1,n)\big\}
		\end{aligned}
		\end{align}
		satisfies $|\mathcal{S}|\in \left[2M/n,6M/n \right]$. If not, they set the output to $R=\text{Fail}$ and abort the protocol.
		\item[3.] They verify that $x_i=y_i$ for every $i\in \mathcal{S}$ and $(u_i,v_i)\neq (1,n)$ or that $x_i\neq y_i$ for $i\in \mathcal{S}$, $(u_i,v_i)=(1,n)$. If any one of these conditions is not satisfied, they set $R=\text{Fail}$ and abort.
		\item[4.] They use further ${\log|S|}$ bits from the $\varepsilon$-SV source to choose $f\in \mathcal{S}$ which indicates the position of the box, from which an output bit $x_f$ is recorded. The protocol outputs $R=x_f$.
	\end{enumerate} 
	\\
	
	\hline
\end{tabular}
\caption{Randomness amplification protocol based on the chained Bell inequality.}
\label{protocol}
\end{figure*}
\vspace{8mm}

\begin{remark}
	\label{remark:m2}
	In Step 1 of the protocol, we require $|\mathcal{S}|\in[2M/n,6M/n]$, since the probability of uniformly choosing neighboring measurement settings is exactly $P(i\in \mathcal{S})=4/n$, for every $i\in\{1,\ldots,M\}$.
\end{remark}
\begin{remark}
	\label{remark:m}
	In the proof we set $M := (n/2)^{2.99}$ and take $n$ such that $\log n$ and $\log M/n$ are integers. We have that $(2M)/n=(n/2)^{1.99}$ and $(6M)/n=3(n/2)^{1.99}$ and the number of boxes labeled by $i\in \mathcal{S}$ is slightly smaller than $(n/2)^2$ (for large $n$). This ensures that the protocol does not abort when run with the optimal quantum strategy while it does abort when run with classical boxes.
\end{remark}

\section{Analysis of the randomness amplification protocol}
\label{sec:analysisOfProtocol}

\subsection{Parameters}

The parameters of the general problem are denoted by $m$, $n$ and $a$. Here $m$ is the number of boxes (runs) in the protocol ($m=|\mathcal{S}|$ in the above protocol based on the chained Bell inequality), $n$ is the number of input pairs that enter the inequality and $a$ is the probability that in any run, a~local box attempting to mimic an ideal box is \textit{not} detected by the measurement.

\subsection{Attacks on the protocol due to the lack of independence}

Consider that an adversary prepares a sequence of boxes of length $m$, and the honest parties obtain bits from the source to input as measurement settings in the runs $i = 1, \dots, m$. In the previously considered scenario in \cite{cr}, the assumption of independence between the source and the device implies that the observation by the honest parties of the ideal sequence of measurement outcomes (i.e., compatible with the optimal violation) guarantees that the true Bell value of the devices used in the protocol is also optimal. Moreover, the distribution of the further bits drawn to choose $f \in \mathcal{S}$ (the position of the box from which the final output bit is drawn) is also independent of the device. Therefore, when the tests in the protocol are passed, the boxes used must be optimal (i.e., as $n \rightarrow \infty$, we have that $\delta^{true} \rightarrow 0$ faster than $1/n$), and perfect randomness may be obtained from the output. 

The relaxation of the independence assumption means that the sequence of boxes supplied by the adversary may be correlated with the bits that the honest parties use in the protocol . This implies that for any given sequence of inputs and corresponding observed outputs $(I = i, O = o)$, there is a class of box sequences that is compatible with this $(i, o)$. We denote such a class in what follows as a "cloud" of box sequences. Moreover, the bit string corresponding to position $f$ is drawn from the same SV source, which means that the SV-condition for boxes in Eq.(\ref{eq:SVforBoxes_S}) applies to it. We will therefore consider attacks limited by the SV-condition as in the following remarks.
\begin{remark}\label{rem1}
	Correlations between measurement settings from the source and boxes are the same as in Sections \ref{sec:theObservedBellValuefromSVcondition} and \ref{sec:exampleTruevsObsChain}, so only the SV-condition for boxes (\ref{eq:SVforBoxes_S}) limits them.
\end{remark}
\begin{remark}\label{rem2}
	We allow attacks in which correlations between sequences of $|\mathcal{S}|$ boxes and the number $f$ are only limited as in Eq.(\ref{eq:SV_cloud}) which follows from the SV- condition for boxes.
\end{remark}

\subsection{{The considered class of attacks and their symmetries}}
\label{sseq:old:6e}

{Within this paper we explore the attacks which consist of box sequences made of extremal boxes for each run. By extremal box we mean either an ideal or a bad one (see Section \ref{sec:chain} for the accurate definitions).} 
{To explain what kind of symmetries occur in the considered class of attacks and to define this class properly, we need to introduce the following notation.} 
We say that a sequence of extremal boxes is of type $j$ if it contains exactly $j$ bad boxes. Let $P_j$ denote the probability of the class of box sequences of type $j$. Obviously,
\begin{equation}
	\label{eq:sum_P_j}
	\sum_{j=1}^mP_j=1.
\end{equation}
{We set the probability that the adversary supplies the box sequence consisting of only ideal boxes to be zero, i.e., $P_0 = 0$, these boxes generate perfect random output over all runs so that using such boxes does not give any advantage to the adversary.} 
Note that within a sequence of $m$ boxes, $j$ bad boxes may be arranged in 
$\binom{m}{j}$ different ways (see Fig. \ref{fig:arrangement}) 

\begin{figure}[!t]
	\centering
		\includegraphics[trim=0cm 4cm 0cm 1cm, width=8cm]{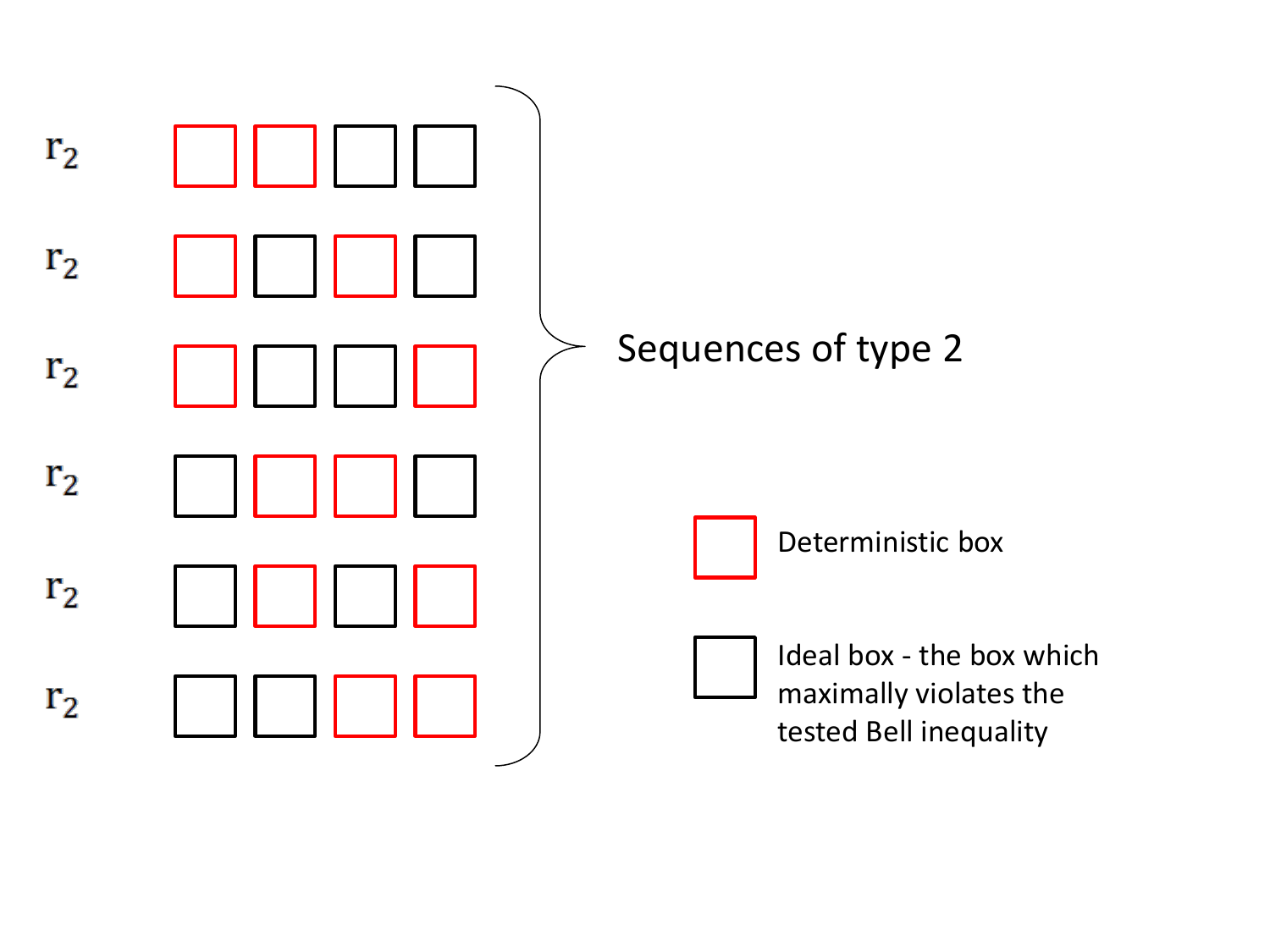}
		\caption{%
			Possible arrangements of 2 bad boxes in a sequence of 4 boxes.
		}
		\label{fig:arrangement}
\end{figure}

{%
	Let us consider the case when any bad box has exactly one contradiction when compared with the correlations in an ideal box. }
\begin{figure*}[!t]
	\centering
	\includegraphics[trim=0cm 2.2cm 0cm 0cm, width=15cm]{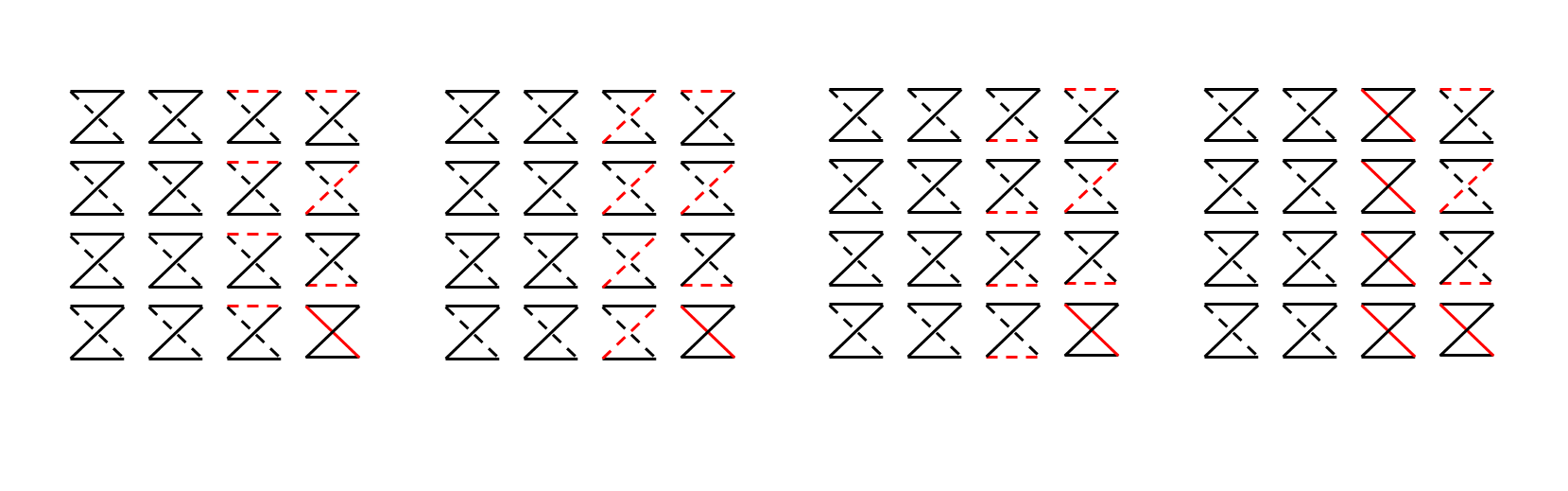}
	\caption{%
		There are $n^j$ sequences of type $j$ and of certain arrangement, e.g. in case of CHSH inequality, 16~different sequences are of type 2 and arrangement: 2 PR-boxes and 2 bad boxes. The edges with mismatched correlations are marked in red.	}
	\label{fig:chain_arrangement}
\end{figure*}
	In this case, there are 
	$\binom{m}{j} n^j$ possible sequences of type $j$ (since the contradiction can happen in any one of the $n$ different measurement pairs, see an example in Fig. \ref{fig:chain_arrangement}). Furthermore, consider the case when every sequence of type $j$ is equally likely, i.e. appears with the same probability~$r_j$, this gives that
\begin{align}
	\label{eq:Pj_as_rj}
	P_j=\binom{m}{j} n^jr_j.
\end{align}


Recall that $f\in\{1,\ldots, m\}$ is the number drawn using bits from the $\varepsilon$-SV source, which indicates the position of a box in a sequence from which the final bit is recorded. Let an arbitrary sequence of type $k$ be denoted by Seq$_k$. 
Then, we consider a family of the attack strategies given by the joint probability of $f$ and all possibly supplied sequences which satisfy, for a given parameter $\lambda \in \left(0,1\right]$, the following condition:
\begingroup\makeatletter\def\f@size{9.5}\check@mathfonts
\begin{align}
	\label{def:attack}
	\begin{aligned}
		P\left(f=i|\text{Seq}_k\right)
		=\left\{
		\begin{array}{lll}
			\frac{\lambda}{k}&\; 
			\text{for $i$}& \text{ being the position number}\\ &&\text{ of bad box in Seq$_k$}\\
			\frac{1-\lambda}{m-k}&\; \text{for $i$}& \text{ being the position number}\\ &&\text{  of ideal boxe in Seq$_k$}.
		\end{array}
		\right.
	\end{aligned}
\end{align}
\endgroup

\begin{figure*}[!t]
	\centering
	\includegraphics[trim=0cm 0.6cm 0cm 0cm, width=10cm]{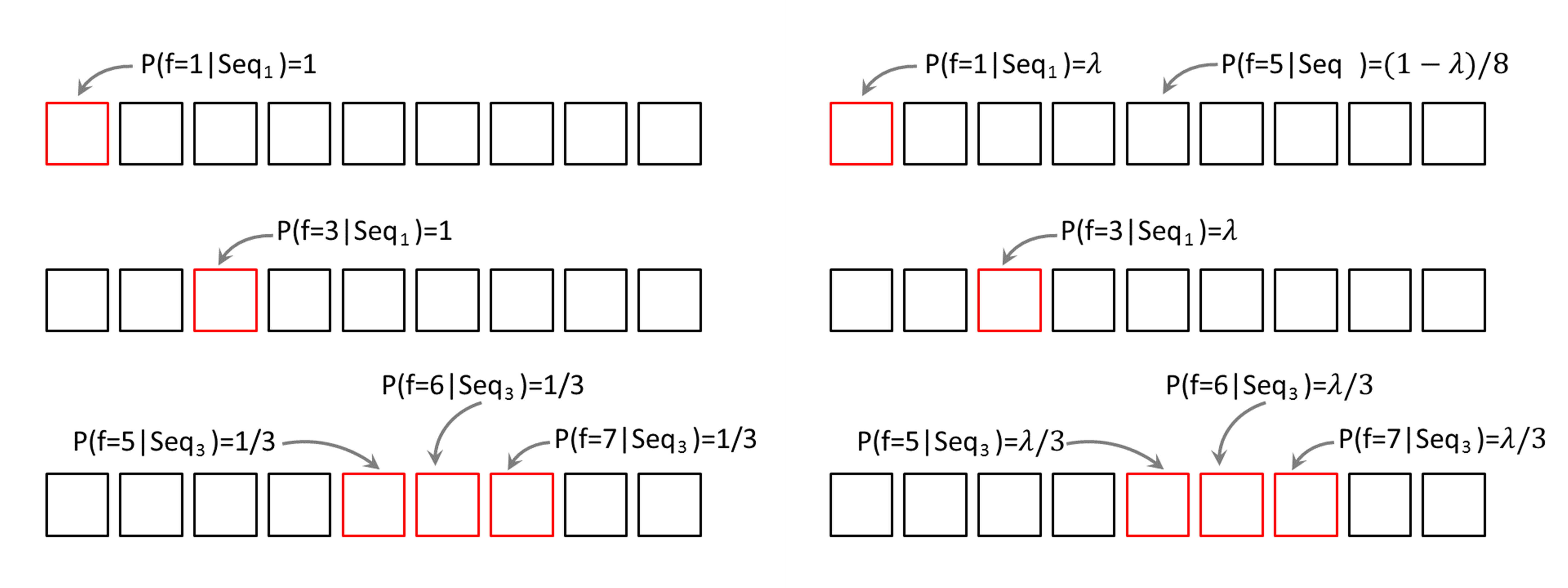}
	\caption{{
		(Left) The probability of $f$ is spread uniformly over bad boxes (cf. Eq. (\ref{def:attack}) with $\lambda=1$). (Right) The attack allows to take the final bit from an ideal box with probability $(1-\lambda)$, which is spread uniformly over all ideal boxes. Moreover, $f$ is distributed uniformly over the bad boxes with probability $\lambda\in(0,1]$ (cf. Eq. (\ref{def:attack}) with an arbitrary parameter $\lambda\in(0,1]$).}	}
	\label{fig:attack}
\end{figure*}

{
\begin{remark}
Note that $P_m> 0$ and $r_m>0$ only in the case when $\lambda=1$. Otherwise, when $\lambda<1$, we set $P_m=r_m=0$, since (according to the definition of an attack) $f$ should indicate the position of an ideal box with probability $1-\lambda$, which is not possible while having no ideal box is a sequence.  
\end{remark}}

Possible attacks are exemplified in Fig. \ref{fig:attack}.

\subsection{Assumptions on the attack strategy}
\label{sseq:old:6f}

We assume that in the attack strategy, any bad box has exactly one contradiction when compared with the correlations in an ideal box. That any attack strategy without this assumption is strictly weaker is justified in Appendix \ref{appendix2}, intuitively it is clear that using local boxes with more contradictions simply decreases the probability of acceptance for the protocol (since the observed Bell value increases) in comparison to using boxes with a single contradiction while yielding the same lack of randomness in the output.  

{After taking the above considerations into account, we end with the following assumptions on the particular class of attacks considered in this paper.}
{
\begin{enumerate}
	\item[1.] We assume that the attack consists of box sequences made of extremal boxes for each run, and defer the consideration of the general attack consisting of a large box coherent over all runs for future work.
	\item[2.] We assume that the attack is symmetric in the sense that every box sequence of particular type $j$ (i.e., containing $j$ bad boxes) appears with the same probability as in Eq.(\ref{eq:Pj_as_rj}).
	\item[3.] We assume that $f$, drawn from the source, is distributed uniformly over the bad boxes with probability $\lambda$ and uniformly over the good boxes with probability $1-\lambda$ for any particular sequence Seq$_k$ as specified in Eq. (\ref{def:attack}). 
\end{enumerate}}

\subsection{Probability of acceptance of the protocol}\label{sec:prob_acc}

Recall that $a$ denotes the probability of not detecting a contradiction with the correlations of an ideal box when measuring a bad box in a single run. Then, the probability of not aborting the protocol, which happens if and only if the correlations in all the runs are compatible with the ideal correlations, is  
described by the following expression:
\begin{equation}
	\label{def:P(ACC)}
	P(\text{ACC})=\sum_{k=1}^mP_ka^k.
\end{equation}
Let us now compute $a$ for the protocol based on the chained inequality. Note that, since only one measurement can be performed, the probability that an edge with contradiction is measured is, in case of uniform and independent inputs, as small as $1/n$ and can be even smaller in the case of inputs taken from the source.
Due to Theorem \ref{thm:ratio}, 
\begin{align}
	n\delta^{\text{true}}\left(\frac{p_{\min}\zeta_{\min}}{p_{\max}}\right)\leq \delta^{\text{obs}},
\end{align}
so that the probability that an edge with contradiction is measured by Alice and Bob is bounded from below by $p_{\min}\zeta_{\min}/p_{\max}$, which in turn implies that
\begin{align}
	a=1-\frac{p_{\min}\zeta_{\min}}{p_{\max}}.
\end{align}
Note that when we consider the probability of not detecting that a subsequent box is local, it is a conditional probability with all proceeding measurements in the condition (see Remark \ref{rem:neglecting_e} in Section \ref{sec:theObservedBellValuefromSVcondition} about an arbitrary random variable $e$ that is prior to the protocol).

{In the rest of the paper, we will show that the protocol stays secure under the class of attacks  described in Section \ref{sseq:old:6f}.

\subsection{Main result}

Let us set $a=(1-p_{\min}\zeta_{\min}/p_{\max})$ (from Section {\ref{sec:prob_acc}}) and $m=|\mathcal{S}|=(n/2)^{1.99}$ (which follows from the requirements of the protocol and the rules of quantum mechanics, see Remark \ref{remark:m}). We approximate terms $p_{\min}$, $p_{\max}$ and $\zeta_{\min}$ as we did in Eq. (\ref{def:p_min/max}).

{
Let us denote by $\mathbf{s}$ the vector of input pairs (which form an edge in a chain), i.e. $\mathbf{s}=\left(s_1,\ldots,s_m\right)$, where $s_i=(u_i,v_i)$ and, according to the protocol, $u_i$ and $v_i$ are drawn form an $\varepsilon$-SV source. Moreover, let us fix an arbitrary value of the variable $f$ taken form an $\varepsilon$-SV source, say $f=f_0$, $f_0\in\{1,\ldots,m\}$. Now, when $f=f_0$ and $\mathbf{s}$ are fixed, the marginal box, from which the output of the protocol shall be obtained, is of the form
\begin{align}
\begin{aligned}
\sum_{\xi_1}\ldots\sum_{\xi_{f_0-1}}\sum_{\xi_{f_0+1}}\ldots\sum_{\xi_m}B_{\mathbf{s},f_0}\left(\xi_1,\ldots,\xi_m,z|\mu_1,\ldots,\mu_m,w\right)
\overset{\text{NS}}{=}B_{\mathbf{s},f_0}^{(f_0)}\left(\xi_{f_0},z|\mu_{f_0},w\right),
\end{aligned}
\end{align}
where $B_{\mathbf{s},f_0}$ is the mixture of extremal box sequences used for the considered class of the attacks (see Sections \ref{sseq:old:6e} and \ref{sseq:old:6f}) and adapted to both $\mathbf{s}$ and $f=f_0$ (as indicated in the indexes), which is the manifestation of the correlations between the $\varepsilon$-SV source and the device. Further, ${\mu}=(\mu_1,\ldots,\mu_m)$ and ${\xi}=(\xi_1,\ldots,\xi_m)$ are the vectors of inputs and outputs of the box, respectively, and $(w,z)$ is the input-output pair of an eavesdropper.  Note that in the case of the honest parties, who follow the protocol given in Section \ref{sec:randomnessAmplificationProtocol}, ${\mu}$ is perfectly correlated with $\mathbf{s}$.

The proof of the following lemma may be found in Appendix \ref{app:lem,f,seq}.
\begin{lemma}\label{lem:f,seq}
Under the assumptions 1-3 outlined in Section \ref{sseq:old:6f}), we obtain
\begin{align}\label{eq:lambda/m}
\begin{aligned}
\sum_{\text{Seq}}P(f=f_0,\text{Seq})\mathbbm{1}_{\{\text{position numbers of det. boxes in Seq}\}}(f_0)=\frac{\lambda}{m},
\end{aligned}
\end{align}
{where Seq denotes an arbitrary sequence of extremal boxes which are product with one another and the sum is over all sequences of this type} ("det." is just the abbreviation for "deterministic"). Let us also indicate that 
\begin{align}\label{eq:1}
\begin{aligned}
\mathbbm{1}_{\{\text{position numbers of det. boxes in Seq}\}}(f_0)=
\left\{\begin{array}{ll}
1,&\text{if the $f_0$-th box in Seq is deterministic}\\
0,&\text{if the $f_0$-th box in Seq is ideal}
\end{array}.\right.
\end{aligned}
\end{align}
As a consequence, we have
\begin{align}\label{eq:f,seq}
\begin{aligned}
&\sum_{f_0=1}^m\sum_{\text{Seq}}P(f=f_0,\text{Seq})\mathbbm{1}_{\{\text{position numbers of det. boxes in Seq}\}}(f_0)\\&=\lambda.
\end{aligned}
\end{align}
{Note that the parameter $\lambda\in(0,1]$ stems from the attack of Eve and determines the average probability (over $f$) of not obtaining a random bit under the protocol given in Section \ref{sec:randomnessAmplificationProtocol}. }
\end{lemma}

\begin{corollary}\label{corol_}
Let $f=f_0$ and $\mathbf{s}$ be arbitrarily fixed. The fraction of boxes which are not ideal (i.e. boxes which do not generate fully or almost fully random bits) {within $B_{\mathbf{s},f_0}^{(f_0)}=\sum_{\text{Seq}}\left(\text{Tr}_{i\neq f_0}\text{Seq}\right)$} is then equal to $\lambda/(mP(f=f_0))$. Indeed, note that Eqs.  (\ref{eq:lambda/m}) and (\ref{eq:1}) immediately imply the following: 
\begin{align}
\sum_{\text{Seq}}P(\text{Seq}\text{ with $f_0$-th det. box}|f=f_0)=\frac{\lambda}{mP(f=f_0)}.
\end{align}
\end{corollary}
}

{
\begin{definition}
Let us define the composable distance $d_c$ between fully random bits and final bits generated within the protocol stated in Section \ref{sec:randomnessAmplificationProtocol} as follows:
\begin{align}\label{def:d_c_prot}
d_c=\sum_{f_0=1}^mP(f=f_0)d_c|_{f_0}\left(B_{\mathbf{s},f_0}^{(f_0)}\right),
\end{align}
where $d_c|_{f_0}(B_{\mathbf{s},f_0}^{(f_0)})$ are determined by Eq. (\ref{def:d_c}) for every $f_0\in\{1,\ldots,m\}$. {Hence $d_c$ denotes the distance in the case when $f$ is known by the distinguisher, i.e. this is the composable distance with public $f$.}
\end{definition}}

Let us now state the definition of the secure protocol.}
\begin{definition}
	We define the probability of error {as the following joint probability:  
	\begin{align}
		P(\text{error}) = P\left(d_c > d^{\text{thr}}(n), P({\text{ACC}}) > \text{ACC}^{\text{thr}}(n)\right),
	\end{align}
	where $d^{\text{thr}}(n)$ and $\text{ACC}^{\text{thr}}(n)$ are some thresholds for $d_c$ and $P(\text{ACC})$, respectively. }
	{We say that a randomness amplification protocol is composably secure  if there exist \( d^{\text{thr}}(n) \) and \( \text{ACC}^{\text{thr}}(n) \), both converging to zero with increasing  $n$}, such that P(\text{error})=0.
\end{definition}

{Let us now state the following crucial lemma.}
\begin{lemma}
	\label{lemma:accBound}
	Assume that the correlations between the source and the device are constrained as in Remarks {\ref{rem1} and \ref{rem2}}. Under the assumptions on the attack strategy outlined in Section \ref{sseq:old:6f}, {the probability of accepting the protocol from Section \ref{sec:randomnessAmplificationProtocol} is estimated as follows:}
	\begin{align}
		P({\text{ACC}}) \leq a^{\lambda/\left(p_+^{\log(m)}\right)}.
	\end{align}
\end{lemma}
Let us postpone the proof of Lemma \ref{lemma:accBound} {to Sections \ref{sec:P(ACC)}-\ref{sec:optimal}} in order to formulate now the main result of this paper. 
\begin{theorem}
	\label{thm:epsilon2}
	Under the assumptions from Lemma \ref{lemma:accBound}, the protocol given in Section \ref{sec:randomnessAmplificationProtocol} is a composably secure randomness amplification protocol for every {private} \( \varepsilon \)-SV source with $\varepsilon < 0.0132$.
\end{theorem}
\begin{proof}
{
First of all, let us note that 
\begin{align}
\begin{aligned}
d_c&\overset{\text{Eq. (\ref{def:d_c})}}{=}\sum_x\max_w\sum_zP(z|w)\left|P(x|z,w)-\frac{1}{2}\right|
\leq 2\max_w\sum_zP(z|w)\sum_x\left|P(x|z,w)-\frac{1}{2}\right|\\
&\leq 2\max_w\sum_zP(z|w)\sum_x\left|\sum_uP(u)P(x|u,z,w)-\frac{1}{2}\right|
\leq 2\max_w\sum_zP(z|w)\sum_x\max_u \left|P(x|u,z,w)-\frac{1}{2}\right|.
\end{aligned}
\end{align}
Therefore, for fixed $\mathbf{s}$ and $f=f_0$ ({recall that $s_i=(u_i,v_i)$ is the pair of inputs of Alice and Bob in the $i$-th run of the protocol, while $x_i$ is the output for Alice in this run}),
we obtain the following estimation:
\begingroup\makeatletter\def\f@size{9.5}\check@mathfonts
\begin{align}
\begin{aligned}
d_c|_{f_0}\left(B_{\mathbf{s},f_0}^{(f_0)}\right)
&\overset{\text{Eq. (\ref{def:d_c})}}{=}\sum_{x_{f_0}}\max_w\sum_zB_{\mathbf{s},f_0}^{(f_0)}(z|w)\left|B_{\mathbf{s},f_0}^{(f_0)}(x_{f_0}|z,w)-\frac{1}{2}\right|\\
&\leq 2\max_w\sum_zB_{\mathbf{s},f_0}^{(f_0)}(z|w)\Big(\max_u\left|B_{\mathbf{s},f_0}^{(f_0)}(x_{f_0}=0|u,z,w)-\frac{1}{2}\right|
+\max_u\left|B_{\mathbf{s},f_0}^{(f_0)}(x_{f_0}=1|u,z,w)-\frac{1}{2}\right|\Big)\\
&\overset{\text{Corollary \ref{corol_}}}{=}2\frac{\lambda}{mP(f=f_0)}\max_w\sum_zB_{\mathbf{s},f_0}^{(f_0)}(z|w)\\
&\qquad\times\Bigg(\max_u\left|P_{bad}(x_{f_0}=0|u,z,w)-\frac{1}{2}\right|
+\max_u\left|P_{bad}(x_{f_0}=1|u,z,w)-\frac{1}{2}\right|\Bigg)\\
&\leq \frac{2\lambda}{mP(f=f_0)}
\end{aligned}
\end{align}
\endgroup
and hence, according to Eq. (\ref{def:d_c_prot}), we have
\begin{align}\label{eq:d_c<}
d_c\leq 2\lambda,
\end{align}
which means that for some small value of $\lambda\in(0,1]$, the composable distance $d_c$ is also small (on average over $f$), while for the higher value of the parameter $\lambda$ no proper upper bound on $d_c$ can be guaranteed.
}

	Lets \( \beta \) be some small positive constant. 
	Then, let 
	\begin{align}
		d^{\text{thr}}(n) := 2n^{-\beta}
	\end{align}
	and
	\begin{align}\label{eq:ACC^thr}
		{\text{ACC}^{\text{thr}}(n) := \left(1-\frac{p_-^{12\log (n/2)}}{np_+^{12\log (n/2)}}\right)^\frac{n^{-\beta}}{p_+^{1.99 \log (n/2)}}}.
	\end{align}
	Obviously, \( \lim_{n \to \infty} d^{\text{thr}}(n) = 0 \). {The second bound converges to zero with increasing $n$ for every $\varepsilon$ satisfying $(0.5-\varepsilon)^{12}-2^{1+\beta}(0.5+\varepsilon)^{13.99}>0$. By putting $\beta$ to zero, we obtain the threshold value for $\varepsilon$, i.e. $0\leq\varepsilon< 0.0132$. 
	
	}
	
	From Lemma \ref{lemma:accBound} we know that $P({\text{ACC}}) \leq a^{\lambda/\left(p_+^{\log(m)}\right)}. $ 
	Lets us choose \( \lambda_0 = n^{-\beta} \).
	There are two possibilities, depending on the attacker's choice of the parameter \( \lambda\in(0,1] \).
	If \( \lambda > \lambda_0 \), then 
	\begin{align}
		P({\text{ACC}}) \leq a^{\lambda/\left(p_+^{\log(m)}\right)} < \text{ACC}^{\text{thr}}(n),
	\end{align}	
	which implies that 
	\begin{align}
		P(\text{error}) = P\left(d_c > d^{\text{thr}}(n), P({\text{ACC}}) > \text{ACC}^{\text{thr}}(n)\right) = 0.
	\end{align}	
	{On the other hand, if \( \lambda \leq \lambda_0 \), then using Eq. \ref{eq:d_c<} we obtain 
	\begin{align}
d_c\leq 2\lambda\leq 2n^{-\beta},
	\end{align}	
	so we also get
	\begin{align}
		P(\text{error}) = P\left(d_c > d^{\text{thr}}(n), P({\text{ACC}}) > \text{ACC}^{\text{thr}}(n)\right) = 0,
	\end{align}	
	which ends the proof.}
\end{proof}
{To complete the security proof, it now remains to establish Lemma \ref{lemma:accBound}. To do this, we first need to introduce some useful notation and the main concepts of our reasoning.}

\subsection{The notion of clouds}

If we measure a bad box, we may either observe a~contradiction with the correlations in an ideal box or not. Not observing a contradiction does not guarantee that the box is ideal. This leads to the notion of clouds, i.e., classes of boxes compatible with a given sequence of observations for a chosen sequence of measurement inputs. If $1$ denotes the event that a contradiction is observed and $0$ denotes the complementary event, the pattern of zeros and ones (of length $m$), together with the chosen sequence of measurement settings, defines the cloud. Let a sequence of measurement settings be fixed. We denote the cloud by $\mathcal{C}^{\mathbf{l}}$, where $\mathbf{l}=(l_1,\ldots,l_m)$ and $l_1,\ldots,l_m\in\{0,1\}$. Note that $|\mathbf{l}|=\sum_{j=1}^ml_j$ delivers information about the number of detected contradictions, hence detected bad boxes. So there are at least $|\mathbf{l}|$ bad boxes in the sequence which has been measured (see Figs. \ref{fig:cloud} and \ref{fig:cloud_chain}). Hence, in every cloud $C^{\mathbf{l}}$ there are boxes of type $q$ for $q \geq |\mathbf{l}|$, but only of certain arrangements, determined by the performed measurements (see Fig. \ref{fig:cloud} for an example set of arrangements).

\begin{figure}[!t]
	\centering
		\includegraphics[trim=0cm 3.2cm 0cm 2cm, width=7.8cm]{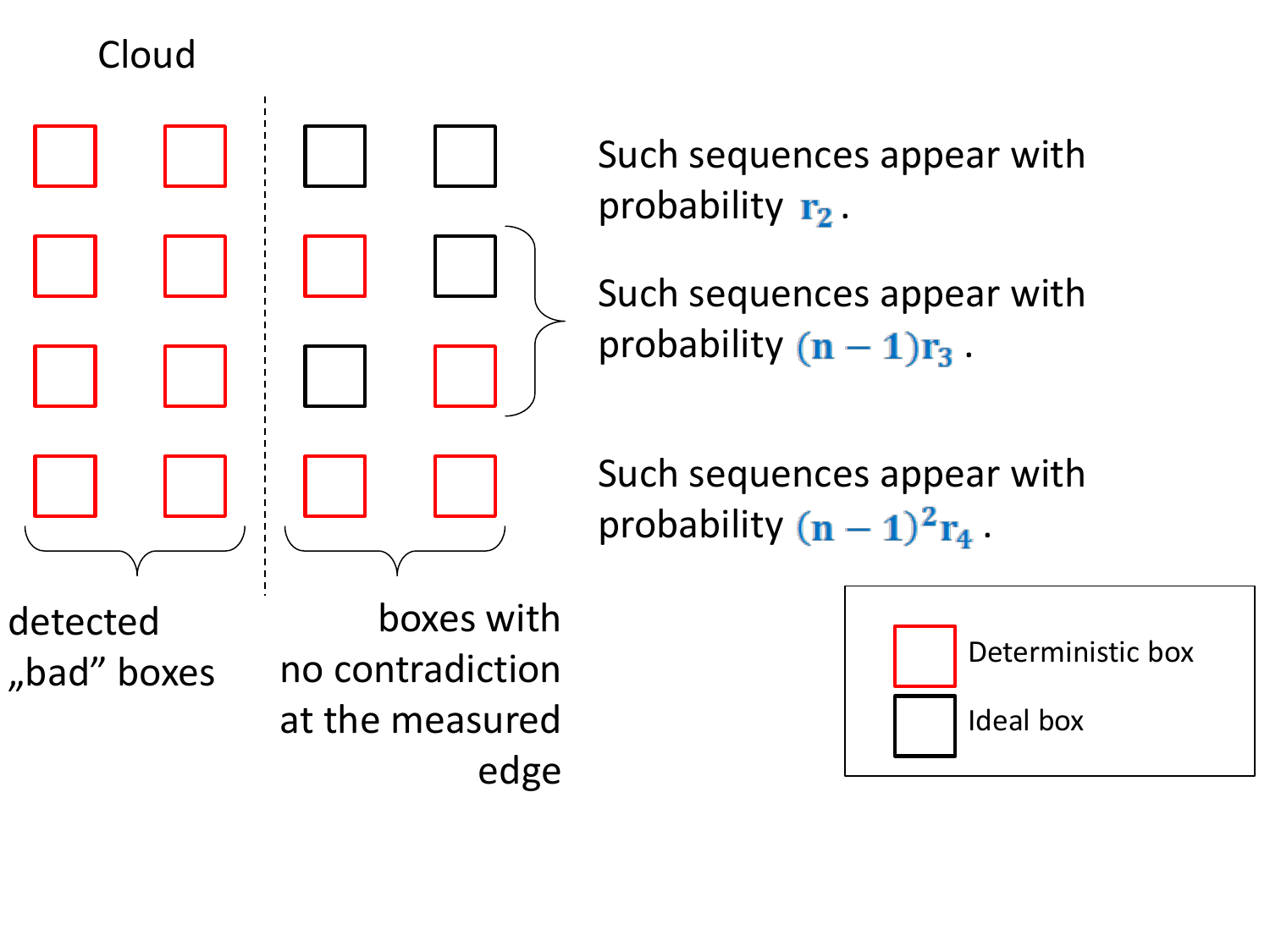}
		\caption{%
			The cloud $\mathcal{C}^{(1,1,0,0)}$ (with 2 detected bad boxes).
			First two boxes are bad, which is known after performing a measurement, the next two may be either ideal or bad boxes.
		}
		\label{fig:cloud}
\end{figure}

\begin{figure}[!t]
	\centering
		\includegraphics[trim=0cm 0.2cm 0cm 0cm, width=8.2cm]{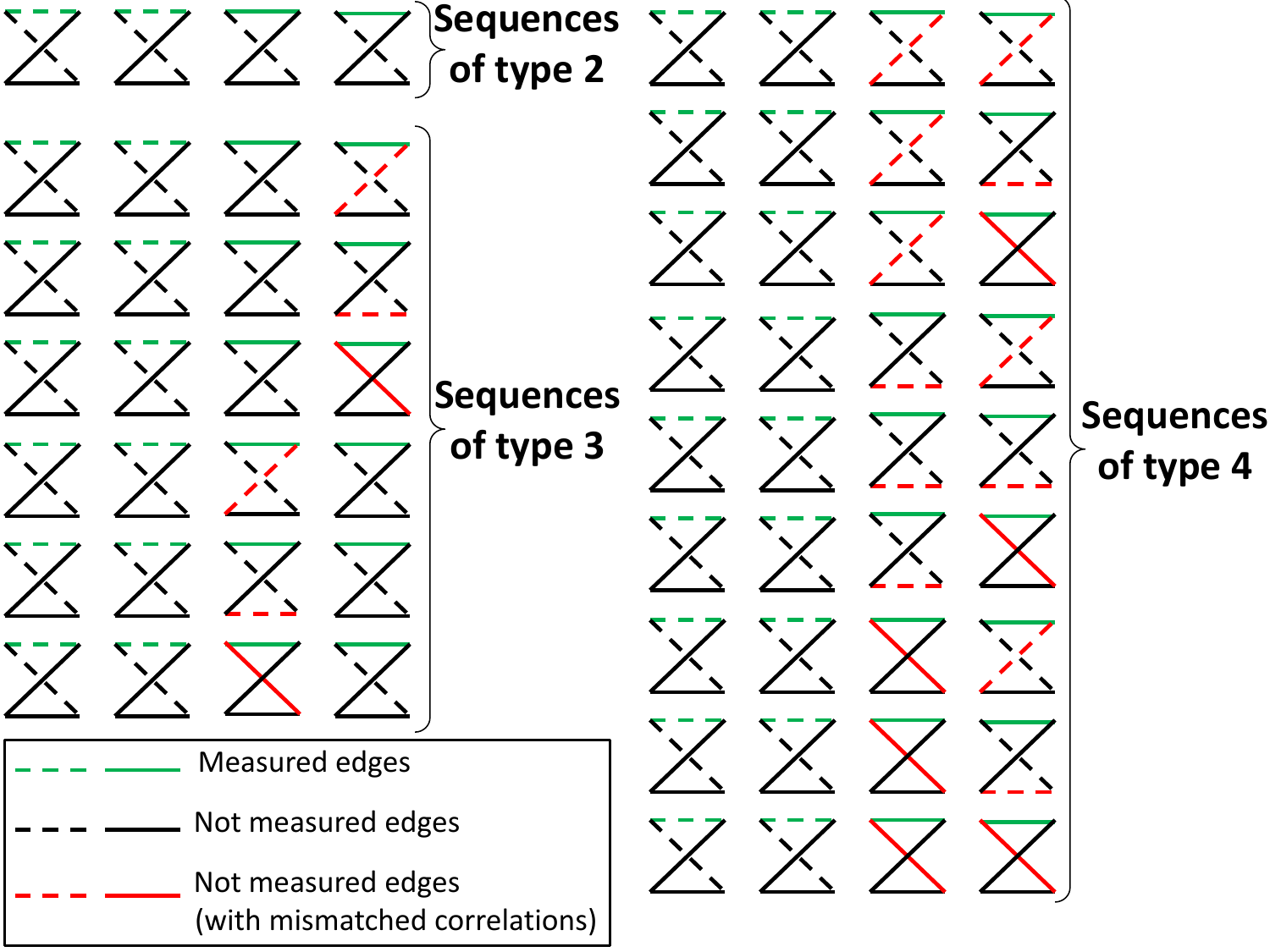}
		\caption{%
			The cloud $\mathcal{C}^{(1,1,0,0)}$  
			in case of CHSH inequality. 
			First two boxes are bad, which is known after performing a measurement, the next two may be either PR-boxes or bad boxes.
		}
		\label{fig:cloud_chain}
\end{figure}

Note that detecting a contradiction gives certainty that the box is bad, as well as the knowledge where exactly the contradiction appears. Not detecting a contradiction delivers only information that there is no contradiction at the certain edge which has been measured. We may not exclude the possibility that there is a contradiction at any other (non-measured) edge (which is also indicated in the example in Fig. \ref{fig:cloud}). 
It should be noted that clouds overlap at each other, i.e., the same sequence of boxes may appear in multiple clouds. 
Let $Q_l=P(\mathcal{C}^{\mathbf{l}})$ for $|\mathbf{l}|=l$. 
Referring to the above analysis {(especially Section \ref{sseq:old:6e})}, we obtain
\begin{equation}
	\label{eq:cloud_prob}
	Q_l=\sum_{s=0}^{m-l}\binom{m-l}{s}(n-1)^s r_{l+s}\quad\text{for }l\in\{1,\ldots,m\}.
\end{equation}
{%
	Note that there are 
	$\binom{m}{l}$ clouds which appear with probability $Q_l$
}.

\subsection{Constraints following from the SV-condition for boxes}

We have that 
\begin{align}
p_-^{\log m} \leq P\left({f=i}|\text{a sequence of measurements} \right) \leq p_+^{\log m},
\end{align}  since $f$ is a bit string 
drawn from the $\varepsilon$-SV source after the bits corresponding to the sequence of measurements are drawn from the same source.
The assumed SV-condition for boxes in Eq.(\ref{eq:SVforBoxes_S}) then implies that 
\begin{equation}
\begin{aligned}
	p_-^{\log m}
	&\leq P\left({f=i}|\text{a sequence of measurements and outcomes}\right)
	\leq p_+^{\log m}
\end{aligned}
\end{equation}
for every $i\in\{1,\ldots,m\}$. 
Note that there is a one-to-one correspondence between the sequence of measurements and outcomes and its corresponding cloud. 
Suppose that measurement settings are fixed and some outcomes are obtained. Then the appropriate cloud $\mathcal{C}^{\mathbf{l}}$ is determined and we have
\begin{equation}
	\label{eq:SV_cloud}
	p_-^{\log m}\leq P\left({f=i}|\mathcal{C}^{\mathbf{l}}\right)\leq p_+^{\log m}\quad\text{for }\;i\in\{1,\ldots,m\}.
\end{equation}
{Let us set $c_+ := p_+^{\log m}/\lambda$ and $k := |\mathbf{l}|$. Since $P$(ACC), given by Eq. (\ref{def:P(ACC)}), is defined in terms of probabilities $P_k$ (see Eq. (\ref{eq:Pj_as_rj})), condition (\ref{eq:SV_cloud}) should also be rewritten in this way. Due to the definition of attack (see Eq. (\ref{def:attack})) and the properties of clouds, we obtain
\begin{align}
	\label{eq:bound}
	\begin{aligned}
		\sum_{s=0}^{m-k}\left(\frac{1}{k+s}-c_+\right)\binom{k+s}{k}\left(\frac{n-1}{n}\right)^sP_{k+s}
		\leq 0.
	\end{aligned}
\end{align}
The derivation of Eq. (\ref{eq:bound}) is given in Appendix \ref{appendix3}.

\subsection{Probability of acceptance as a linear program}\label{sec:P(ACC)}

The probability of acceptance can therefore be formulated as the following linear program. We want to maximize the expression
\begin{equation}
	\sum_{k=1}^mP_ka^k
\end{equation}
such that 
\begin{align}
	\begin{aligned}
		\sum_{s=0}^{m-k}\left(\frac{1}{k+s}-c_+\right)\binom{k+s}{k}\left(\frac{n-1}{n}\right)^sP_{k+s}
		\leq 0
	\end{aligned}
\end{align}
for every $k\in\{1,\ldots,m\}$, and
\begin{equation}
	\sum_{k=1}^mP_k\leq 1,\qquad
	\sum_{k=1}^m-P_k\leq -1,
\end{equation}
where the problem constraints follow from Eqs. (\ref{eq:bound}) and {(\ref{eq:sum_P_j})}. Obviously, 
\begin{equation}
	P_k\geq 0\qquad\text{for every }k\in\{1,\ldots,m\}.
\end{equation}
Note that the linear program written above is at once in its standard form, that is
\begin{align}
	\begin{aligned}
		&\max \left\{\vec{c}^T\vec{x}\right\}\\
		&\text{such that}\quad A\vec{x}\leq \vec{b}\\
		&\text{and the variables are non-negative}\quad \vec{x}\geq 0,
	\end{aligned}
\end{align}
where $\vec{x}=\left(P_1,\ldots,P_m\right)^T$, $\vec{c}=\left(a,a^2,\ldots,a^m\right)^T$, $\vec{b}=\left(0,\ldots,0,1,-1\right)^T$ 
and $A$ is an $(m+2)\times m$ matrix of the form (\ref{matrix}).
\begin{figure*}[!t]
\begin{align}\label{matrix}
	\begin{aligned}
		A=\left[
		\begin{array}{ccccc}
			\binom{1}{0}\left(\frac{n-1}{n}\right)^0\left(1-c_+\right)
			&\binom{2}{1}\left(\frac{n-1}{n}\right)^1\left(\frac{1}{2}-c_+\right)
			&\binom{3}{2}\left(\frac{n-1}{n}\right)^2\left(\frac{1}{3}-c_+\right)
			&\ldots
			&\binom{m}{m-1}\left(\frac{n-1}{n}\right)^{m-1}\left(\frac{1}{m}-c_+\right)\\
			0
			&\binom{2}{0}\left(\frac{n-1}{n}\right)^0\left(\frac{1}{2}-c_+\right)
			&\binom{3}{1}\left(\frac{n-1}{n}\right)^1\left(\frac{1}{3}-c_+\right)
			&\ldots
			&\binom{m}{m-2}\left(\frac{n-1}{n}\right)^{m-2}\left(\frac{1}{m}-c_+\right)\\
			\vdots&\vdots&\vdots&\vdots&\vdots\\
			0&0&0&\ldots&\binom{m}{0}\left(\frac{n-1}{n}\right)^0\left(\frac{1}{m}-c_+\right)\\
			1&1&1&\ldots&1\\
			-1&-1&-1&\ldots&-1
		\end{array}
		\right].
	\end{aligned}
\end{align}
\end{figure*}

\subsection{Dual problem}
{Let us consider the following dual problem:}
\begin{align}
	\begin{aligned}
		\left\{%
		\begin{array}{ll}
			&\min \{\vec{b}^T\vec{y}\}\\
			&A^T\vec{y}\geq \vec{c}\\
			&\vec{y}\geq 0.
		\end{array}
		\right.
	\end{aligned}
\end{align}
By linear programming duality, if either the primal or dual has an optimal solution, then both have optimal solutions and the optimal values of the objective functions of these problems are equal.

In our case the dual problem is as follows:
\begin{align}
	\label{dual}
	\min \{y_{m+1}-y_{m+2}\}
\end{align}
{with constraints}
\begin{equation}
	\label{constraint_dual}
	\sum_{r=0}^{k-1}\binom{k}{r}\left(\frac{n-1}{n}\right)^r\left(\frac{1}{m}-c_+\right)y_{k-r}+y_{m+1}-y_{m+2}\geq a^k
\end{equation}
for $k\in\{1,\ldots,m\}$, 
and 
\begin{equation}\label{dual_last}
	y_1\geq 0,\quad \ldots,\quad y_m\geq 0,\quad y_{m+1}\geq 0, \quad y_{m+2}\geq 0.
\end{equation}
We find the following feasible solution to the dual problem, formulated as Lemma \ref{lem:feasibility}, and proven in Appendix \ref{appendix4}:
\begin{align}
	\label{hypothesis}
	\begin{aligned}
		y_1=\frac{a^{1/c_+}(1-a)}{\left(\frac{1}{c_+}+1\right)\left(\frac{n-1}{n}\right)^{1/c_+}},\qquad
		y_2=y_3=\ldots=y_m=y_{m+2}=0,\qquad
		y_{m+1}=a^{1/c_+}.
	\end{aligned}
\end{align}
\begin{lemma}
	\label{lem:feasibility}
	Hypothesis (\ref{hypothesis}) gives {the} feasible solution of {the} dual problem described by Eqs. {(\ref{dual})-(\ref{dual_last})}.
\end{lemma}

\subsection{The optimal solution}\label{sec:optimal}

In fact, Eq.(\ref{hypothesis}) is not only a bound on the probability of acceptance but {it is just} an optimal solution to the linear program. To prove that the above solution is optimal, we will show that the objective functions of both, primal and dual, problems are equal. 

Suppose that the solution of the primal problem is given by 
\begin{align}
\begin{aligned}
	P_u=\frac{1}{(1+s(u,v))},\qquad
	P_v=\frac{s(u,v)}{(1+s(u,v))},\qquad
	P_k=0\quad\text{for }k\notin\{u,v\},
\end{aligned}
\end{align}
where 
\begin{align}
	s(u,v)=\frac{un\left(c_+-1/u\right)}{v(n-1)\left(1/v-c_+\right)}>0
\end{align}
for $u\leq 1/{c_+}\leq v$. 
If we set
\begin{align}
	u=\frac{1}{c_+},\qquad v=\frac{1}{c_+}+1,
\end{align}
we obtain $P_{1/c_+}=1$ and $P_{(1/c_++1)}=0$ and therefore
\begin{align}\label{primal=dual}
	\max\left\{\sum_{k=1}^mP_ka^k\right\}=a^{1/c_+}=\min\left\{y_{m+1}-y_{m+2}\right\},
\end{align}
which indicates that the solution is indeed optimal.
However, we should note that to be more accurate, we should take $u$ and $v$ as natural numbers, i.e.
\begin{align}
	u=\left\lfloor\frac{1}{c_+}\right\rfloor,\qquad v=\left\lfloor\frac{1}{c_+}\right\rfloor+1.
\end{align}
{Note that, referring to Eq. (\ref{primal=dual}) and the definition of $c_+$, the proof of Lemma \ref{lemma:accBound} is now complete.}

\section{Conclusion}\label{sec:conclusion}

We have studied the protocol of Colbeck and Renner \cite{cr} under relaxed assumptions which allow for correlations between the Santha-Vazirani source with the devices used in the protocol. We have proven, that in spite of {such attacks}, a non-zero range of parameter of $\varepsilon$-SV source allows for randomness amplification in the asymptotic limit of a large number of settings.
More precisely, the protocol (see Section \ref{sec:randomnessAmplificationProtocol}) is {composably secure} for a restricted range of $\varepsilon$ even if we admit
\begin{itemize}
	\item[(1)] correlations between measurement settings and devices, only limited by the SV-condition for boxes (see Sections \ref{sec:theObservedBellValuefromSVcondition} and \ref{sec:exampleTruevsObsChain}),
	\item[(2)] attacks such that, {with probability equal to $\lambda\in(0,1]$ (describing the strength of the attack)}, $f$~is  pointing to local boxes, i.e. boxes with no intrinsic randomness (correlations of sequences of boxes with $f$ are only limited by condition (\ref{eq:SV_cloud})
	.
\end{itemize}
{The detailed assumptions on the attack strategy are listed in Section \ref{sseq:old:6f}. First of all, the device used for the protocol is given as a mixture of sequences of boxes which are extremal and product with one another. Moreover, we assume that all the attacks exhibit a certain kind of symmetry, i.e. sequences of the same type (so with the same number of local boxes) are treated equally by the adversary, further any local box within a single sequence of extremal boxes can be pointed by $f$ with the same probability}. 
{It is plausible, based on the experience gained while working with the SV-condition for boxes, that the attack with certain symmetry conditions assumed within this paper is in fact optimal.} Nevertheless, it is not yet formally proven that we can admit the symmetry assumptions in the attack without loss of generality. This is the aim for future work {(cf. \cite{maciek})}. Another interesting line of research, which is already in progress, aims to determine whether the attack can be physically performed or not, i.e., whether the correlations between the weak source and the devices can be created by the adversary physically without breaking the SV condition at this stage. Finally, an important open question is whether the techniques used in this paper can be generalized to relax the assumption of independence in the finite-device protocols of \cite{brghhh}, \cite{ravi} so as to obtain randomness amplification for the entire range of $\varepsilon$, while tolerating a constant level of noise. 

\section*{Appendix I}
\label{appendix1}
{Let us derive the so-called backward SV-condition, determining that from a given bit of SV alone, one can not guess perfectly any of of the bits that are preceding it in time order.}  
Suppose that $A$ and $B$ are some portions of bits from an $\varepsilon$-SV source of the same length $|A|=|B|$. Fix $\bar{a},\bar{b}\in\mathcal{I}$, {where $\mathcal{I}$ is the set of possible measurement settings in the Bell expression}. We assume that the probability we consider is normalized, i.e. $\sum_{a\in\mathcal{I}}P(A=a)=1$. 
Let us prove that condition
\begin{align}
	p_{\min}\leq P(B=\bar{b}|A=\bar{a})\leq p_{\max}
\end{align}
implies that 
\begin{align}
	\label{def:zeta_min_max}
	\zeta_{\min}\leq P(A=\bar{a}|B=\bar{b})\leq \zeta_{\max},
\end{align}
where 
\begin{align}
	\zeta_{\min}=\frac{p_{\min}^2}{|\mathcal{I}|p_{\max}^2}\quad\text{and}
	\quad \zeta_{\max}=1-(|\mathcal{I}|-1)\zeta_{\min}.
\end{align}
Note that the definition of an $\varepsilon$-SV source (\ref{SV}) implies that
\begin{align}
	P(A=\bar{a},B=\bar{b})=P(A=\bar{a})P(B=\bar{b}|A=\bar{a})\geq p_{min}^2.
\end{align}
Let us now estimate
\begin{align}
	P(B=\bar{b})=\sum_{a\in\mathcal{I}} P(A=a,B=\bar{b})
	\leq p_{\max}^2 |\mathcal{I}|.
\end{align}
We obtain
\begin{align}
	P(A=\bar{a}|B=\bar{b})=\frac{P(A=\bar{a},B=\bar{b})}{P(B=\bar{b})}\geq \frac{p_{\min}^2}{|\mathcal{I}|p_{\max}^2},
\end{align}
which proves the left side of Eq. (\ref{def:zeta_min_max}). 
The formula for $\zeta_{\max}$ may be justified as follows:
\begin{align}
\begin{aligned}
	P(A=\bar{a}|B=\bar{b})
	&=1-\sum_{a\in\mathcal{I}\backslash \{\bar{a}\}}{P(A={a}|B=\bar{b})}\\
	&\leq 1- \zeta_{\min}(|\mathcal{I}|-1).
	\end{aligned}
\end{align}
\label{appendix5}

 
 
 
 
Let us restate the assumptions in the context of the chained Bell inequality:
\begin{enumerate}
 	\item Alice and Bob are spatially separated and share a no-signaling box with two input sets of size $n/2$ and two binary outputs, which violates the chained Bell inequality up to $\delta_Q$. They choose their settings, each using $r = \log(n/2)$ bits from the main part of the $\varepsilon$-SV source ($n$ is taken to be an appropriate integer of the form $2^{r+1}$), i.e., the variable $I_{HP}$ describing their inputs, is perfectly correlated with $S$ as in Eq.(\ref{eq:P(I|S)}) . 
 	\item The SV-condition for boxes (\ref{eq:SVforBoxes_S}) is satisfied with $p_{\min}, p_{\max}, \zeta_{\min}$ given by (\ref{def:p_min/max}). 
 	\item The main part of the source is correlated with the device used by Alice and Bob. Another part, called SV$_{\text{test}}$, is not directly correlated with a device, it is only used to check whether the SV-condition for boxes is violated (details are described in Section \ref{sseq:old:3e}).
\end{enumerate}
 
\begin{theorem}
	\label{thm:epsilon1}
 	Assume that conditions 1-3 are satisfied. 
 	Then, $\varepsilon< (2^{1/12}-1)/(2(2^{1/12}+1)))$ ($\approx 0.0144$) guarantees full randomness of the output in the asymptotic scenario of a large number of inputs $n \rightarrow \infty$.
\end{theorem}

\begin{proof}
 	{Note that Eq. (\ref{eq:dc<deltaQ}) in Corollary \ref{corol} immediately implies that, to verify that output bits are fully random ($d_c{(B_{\text{final}})}\to 0$), it is enough to show that }
 	\begin{align}
	 	\Delta := \delta_Q  \frac{2p_{\max}}{p_{\min}}\zeta_{\min}\to 0,\qquad\text{as }n\to\infty.
 	\end{align}
 	Following Eqs. (\ref{def:delta_Q}), (\ref{def:p_min/max}), we obtain
 	\begin{align}
 		\begin{aligned}
	 		\Delta&=
		 	{2}\sin^2\left(\frac{\pi}{2n}\right)\frac{p_{\max}}{p_{\min}\frac{p_{\min}^2}{np_{\max}^2}}
	 		\leq {2}\left(\frac{\pi}{2n}\right)^2 \frac{np_{\max}^3}{p_{\min}^3}
		 	=\left(\frac{\pi^2}{{2}}\right)\frac{1}{n}\frac{p_{\max}^3}{p_{\min}^3}\\
	 		&=\left(\frac{\pi^2}{{2}}\right)\frac{1}{n}
		 	\frac{p_+^{6r}}{\left(p_+^{2r}+(n-1)p_-^{2r}\right)^3}
	 		\frac{n^3p_+^{6r}}{p_-^{6r}}
		 	=\left(\frac{\pi^2}{{2}}\right)\frac{n^2p_+^{12r}}{p_-^{6r}\left(p_+^{2r}+(n-1)p_-^{2r}\right)^3}.
	 	\end{aligned}
 	\end{align}
 	Setting $n=2^{r+1}$, we have
 	\begin{align}
 		\begin{aligned}
	 		\Delta=\left(\frac{\pi^2}{2}\right)\frac{4^{r+1}p_+^{12r}}{p_-^{6r}\left(p_+^{2r}+\left(2^{r+1}-1\right)p_-^{2r}\right)^3}.
	 	\end{aligned}
 	\end{align}
 	Let us now consider the asymptotic scenario of a large number of settings $r \rightarrow \infty$,
 	\begin{align}
 		\begin{aligned}
 			\lim_{r\to\infty}
	 		\frac{4^{r+1}p_+^{12r}}{p_-^{6r}\left(p_+^{2r}+\left(2^{r+1}-1\right)p_-^{2r}\right)^3}
	 		=0,
 		\end{aligned}
 	\end{align}
 	which imposes that $\varepsilon$ is bounded as
 	\begin{align}
 		\varepsilon<\frac{2^{1/12}-1}{2\left(2^{1/12}+1\right)} \approx 0.0144.
 	\end{align} 
 	Therefore, for the range $0 \leq \varepsilon < (2^{1/12}-1)/(2(2^{1/12}+1))$, we obtain a random output in the asymptotic scenario of a large number of inputs.
 \end{proof}
 
 \begin{remark}	
 	The threshold is in fact slightly bigger (precisely it is $(2^{1/6(2-c)} - 1)(2(2^{1/6(2-c)} + 1)) \approx 0.0162$ where $c$ solves $H(c/2) = 1/2$ for the binary entropy $H$), which can be proven with more accurate approximations for $p_{\min}$, $p_{\max}$ and $\zeta_{\min}$, obtained by using the Ky Fan norm (see \cite{ghhhpr}), i.e., in the regime of large $n$ 
  	\begin{align}
 	\begin{aligned}
	 	p_{min} = \frac{p_-^{2r}}{p_-^{2r} + 2^r p_+^{(2-c)r} p_-^{cr}}, \;
 		p_{max} = \frac{p_+^{2r}}{p_{+}^{2r} + 2^r p_-^{(2-c)r} p_{+}^{cr}} .
 	\end{aligned}
 	\end{align}
 	
 
\end{remark}

\section*{Appendix 2}
\label{appendix2}
Let us justify that to prove that the protocol is safe it is enough to consider boxes with either zero or one contradiction with the correlations of ideal boxes. It should be noted that using bad boxes with more than one contradiction simply decreases the probability of acceptance $P$(ACC) for the protocol, making the observed Bell value bigger. We now show that the attack with bad boxes possessing more than one contradiction can be improved by replacing these boxes with $1$-contradiction boxes. There is now only one more issue that needs attention. Due to the symmetry assumption, on which our analysis is based, we need to replace boxes in such a way, that the final ensemble is symmetric. Fortunately, it can be easily achieved. Indeed, suppose that any box with $k$ contradictions on edges $e_1,\ldots,e_k$ is replaced (with probability $1/k$) by one of boxes with exactly one contradiction at one of edges $e_1,\ldots,e_k$. Then, if we assume that all boxes with $k$ contradictions are equally likely and are treated as described above, we will obtain the symmetric ensemble used in the main text, which justifies that constraints used in linear programming remain the same.

\section*{Appendix 3}
\label{app:lem,f,seq}
{\begin{proof}[Proof of Lemma \ref{lem:f,seq}]
Let us write shortly $\mathbbm{1}_{\{\text{p.n.d.b.Seq}\}}$ for the 
function $\mathbbm{1}_{\{\text{position numbers of det. boxes in Seq}\}}$. Note that 
\begin{align}
\begin{aligned}
\sum_{\text{Seq}}P(f=f_0,\text{Seq})\mathbbm{1}_{\{\text{p.n.d.b.Seq}\}}(f_0)
=
\sum_{k=1}^{m}\sum_{\text{Seq$_k$}}P(f=f_0,\text{Seq}_k)\mathbbm{1}_{\{\text{p.n.d.b.Seq$_k$}\}}(f_0),
\end{aligned}
\end{align}
where Seq$_k$ denotes an arbitrary sequence of type $k$ (see Section \ref{sseq:old:6e}, where the structure of the considered class of the attacks is explained). Let us also recall that, according to our assumptions summarized in Section \ref{sseq:old:6f}, every sequence of type $k$ is equally likely, i.e. appears with the same probability~$r_k$. We further obtain
\begingroup\makeatletter\def\f@size{9.5}\check@mathfonts
\begin{align}
\begin{aligned}
\sum_{k=1}^{m}\sum_{\text{Seq$_k$}}P(f=f_0,\text{Seq}_k)\mathbbm{1}_{\{\text{p.n.d.b.Seq}\}}(f_0)
&=\sum_{k=1}^{m}\sum_{\text{Seq$_k$ with det. $f_0$-th box}}P(f=f_0|\text{Seq}_k)P(\text{Seq}_k)\\
&\overset{\text{Eq. (\ref{def:attack})}}{=}\sum_{k=1}^{m}\sum_{\text{Seq$_k$ with det. $f_0$-th box}}\frac{\lambda}{k}r_k
=\sum_{k=1}^{m}\frac{\lambda}{k}r_k{m-1\choose k-1}n^k,
\end{aligned}
\end{align}
\endgroup
where there are ${m-1\choose k-1}$ arrangements of $k$ deterministic boxes within a sequence of $m$ boxes, when one of them has the already fixed position (in the $f_0$-th place). Finally, we obtain
\begin{align}
\begin{aligned}
\sum_{k=1}^{m}\frac{\lambda}{k}r_k{m-1\choose k-1}n^k
\overset{\text{Eq. (\ref{eq:Pj_as_rj})}}{=}\sum_{k=1}^{m}\frac{\lambda}{k}\frac{P_k}{{m\choose k}}{m-1\choose k-1}
=\frac{\lambda}{m}\sum_{k=1}^{m}P_k\overset{\text{Eq. (\ref{eq:sum_P_j})}}{=}\frac{\lambda}{m},
\end{aligned}
\end{align}
which completes the proof Eq. (\ref{eq:lambda/m}) and hence also the proof of Lemma \ref{lem:f,seq}.
\end{proof}}

\section*{Appendix 4} 
\label{appendix3}

Here, we derive {certain} constraints on the linear program, {given in} Eq. (\ref{eq:bound}) {(see Section \ref{sec:P(ACC)} for the complete formulation of this linear program)}. 
Recall that $k := |l|$. Let us introduce 
disjoint sets $T_{k+s}$, $s\in\{0,\ldots,m-k\}$, such that $\bigcup_{s=0}^{m-k}T_{k+s}=\mathcal{C}^{\mathbf{l}}$. Every set $T_{k+s}$ consists of sequences with $k+s$ bad boxes and belongs to the cloud $\mathcal{C}^{\mathbf{l}}$, which simply means that { the positions of} $k$ detected bad boxes (with contradictions on measured edges) are {fixed}. Note that 
\begin{align}
	\label{eq:|T|}
	|T_{k+s}|=\binom{m-k}{s}(n-1)^s.
\end{align}
We now obtain
\begingroup\makeatletter\def\f@size{9.5}\check@mathfonts
\begin{align}
	\begin{aligned}
		P\left(f=i|\mathcal{C}^{\mathbf{l}}\right)=\frac{P(f=i,\mathcal{C}^{\mathbf{l}})}{Q_k}
		=\frac{1}{Q_k}\sum_{s=0}^{m-k}P(f=i,T_{k+s})
		=\frac{1}{Q_k}\sum_{s=0}^{m-k}\sum_{\text{Seq}_{k+s}\in T_{k+s}}P(f=i|\text{Seq}_{k+s})P(\text{Seq}_{k+s}).
	\end{aligned}
\end{align}
\endgroup
{
	Let us assume that $i$ is defining the position of some detected bad box, which means that $i$ is defining the position of a bad box in every Seq$_{k+s}$ belonging to cloud~$\mathcal{C}^{\mathbf{l}}$. 
	Following the definition of the attack (see Eq. (\ref{def:attack})), as well as Eq. (\ref{eq:|T|}), we obtain
\begingroup\makeatletter\def\f@size{9.5}\check@mathfonts	
	\begin{align}
		\begin{aligned}
			P\left(f=i|\mathcal{C}^{\mathbf{l}}\right)
			&=\frac{1}{Q_k}\sum_{s=0}^{m-k}\frac{\lambda}{k+s}\sum_{\text{Seq}_{k+s}\in T_{k+s}}r_{k+s}
			=\frac{1}{Q_k}\sum_{s=0}^{m-k}\frac{\lambda}{k+s}r_{k+s}|T_{k+s}|
			=\frac{1}{Q_k}\sum_{s=0}^{m-k}\frac{\lambda}{k+s}r_{k+s}\binom{m-k}{s}(n-1)^s.
		\end{aligned}
	\end{align}
	\endgroup
}

We further obtain (due to Eqs. (\ref{eq:SV_cloud}) and (\ref{eq:cloud_prob}))
\begingroup\makeatletter\def\f@size{9.5}\check@mathfonts
\begin{align}
	\begin{aligned}
		\sum_{s=0}^{m-k}&\frac{\lambda}{k+s}\binom{m-k}{s}(n-1)^sr_{k+s}
		\leq p_+^{\log m} Q_k
		=\sum_{s=0}^{m-k}p_+^{\log m}\binom{m-k}{s}(n-1)^sr_{k+s},
	\end{aligned}
\end{align}
\endgroup
 which gives
\begingroup\makeatletter\def\f@size{9.5}\check@mathfonts
\begin{align} 
	\begin{aligned}
		\sum_{s=0}^{m-k}\left(\frac{\lambda}{k+s}-p_+^{\log m}\right)\binom{m-k}{s}(n-1)^sr_{k+s}
		\leq 0.
	\end{aligned}
\end{align}
\endgroup
Then, according to the definition of $P_j$ (see Eq. (\ref{eq:Pj_as_rj})), we have
\begingroup\makeatletter\def\f@size{9.5}\check@mathfonts
\begin{align}
	\begin{aligned}
		\frac{1}{\binom{m}{k}n^k}\sum_{s=0}^{m-k}\left(\frac{\lambda}{k+s}-p_+^{\log m}\right)\binom{k+s}{k}\frac{(n-1)^s}{n^s}P_{k+s}
		\leq 0.
	\end{aligned}
\end{align}
\endgroup
Finally we obtain
\begingroup\makeatletter\def\f@size{9.5}\check@mathfonts
\begin{align}
	\begin{aligned}
		\sum_{s=0}^{m-k}\left(\frac{1}{k+s}-\frac{p_+^{\log m}}{\lambda}\right)\binom{k+s}{k}\left(\frac{n-1}{n}\right)^sP_{k+s}
		\leq 0.
	\end{aligned}
\end{align}
\endgroup

\section*{Appendix 5}
\label{appendix4}
\begin{proof}[Proof of Lemma \ref{lem:feasibility}] 
	{
		To show feasibility, we need to prove that all $m$ inequalities, given by Eq. (\ref{constraint_dual}), are satisfied. The proof falls into three steps.
	}
\begin{itemize}
	\item[I.]
	Let $u\leq v$. Suppose that constraints (\ref{constraint_dual}) for $k=u$ and $k=v$ are equalities. Then, since $y_2=y_3=\ldots=y_n=y_{n+2}=0$, we have
	\begin{align}
		\label{eq:constr_y}
		\begin{aligned}
			&u\left(\frac{n-1}{n}\right)^{u-1}\left(\frac{1}{u}-c_+\right)y_1+y_{m+1}=a^u,\\
			&v\left(\frac{n-1}{n}\right)^{v-1}\left(\frac{1}{v}-c_+\right)y_1+y_{m+1}=a^v.
		\end{aligned}
	\end{align}
	Suppose that 
	\begin{align}
		u=\frac{1}{c_+}\quad\text{and}\quad v=\frac{1}{c_+}+1.
	\end{align}
	Then, after subtracting Eqs. (\ref{eq:constr_y}), we obtain
	\begin{align}
		y_1=\frac{a^{1/c_+}-a^{1/c_++1}}{c_+\left(\frac{n-1}{n}\right)^{1/c_+}}\geq 0.
	\end{align}
	Further, we verify the remaining constraints:
	\begin{align}
		\label{eq:constraints_k}
		\frac{k\left(\frac{n-1}{n}\right)^{k-1}\left(\frac{1}{k}-c_+\right)a^{1/c_+}(1-a)}{c_+\left(\frac{n-1}{n}\right)^{1/c_+}}+a^{1/c_+}\geq a^{k}
	\end{align}
	\item[II.]
	Take $k<1/c_+$ and set $0<l=1/c_+-k$. Then $k\left(1/k-c_+\right)=1-kc_+=lc_+$ and we may write Eq.~(\ref{eq:constraints_k}) as follows:
	\begin{align}
		\label{eq:constraint_to_prove_1}
		\frac{l(1-a)}{\left(\frac{n-1}{n}\right)^{l+1}}+1\geq a^{-l}.
	\end{align}
	To justify that this is true, we carry out the following reasoning. First, note that
	\begin{align}
		\label{eq:p_min<1/n}
		(1-a)\leq \frac{1}{n},
	\end{align}
	which follows from the fact that the minimal biased probability always is lower than the uniform one. Hence, we obtain
	\begin{align}
		\label{eq:relation_a_n}
		a^{-l}\leq \left(\frac{n-1}{n}\right)^{-(l+1)}.
	\end{align}
	Now, it is enough to prove that
	\begin{align}
		\label{eq:to_prove_almost_final1}
		l(1-a)a^{-l}+1\geq a^{-l},
	\end{align}
	since it implies Eq. (\ref{eq:constraint_to_prove_1}), due to Eq. (\ref{eq:relation_a_n}). Let us write Eq. (\ref{eq:to_prove_almost_final1}) as follows:
	\begin{align}
		\label{eq:to_prove_final1}
	l(1-a)+a^l-1\geq 0.
	\end{align}
	We have
	\begin{align}
		\label{eq:first_derivative}
		\frac{d}{dl}\left(l(1-a)+a^l-1\right)=(1-a)+a^l\ln(a),
	\end{align}
	where $\ln$ is the natural logarithm. Note that, since $\ln(a)<0$, we have
	\begin{align}
		(1-a)+a^l\ln(a)\geq (1-a)+a\ln(a). 
	\end{align}
	Let us verify if
	\begin{align}
		(1-a)+a\ln(a) \geq 0,
	\end{align}
	which is equivalent to
	\begin{align}
		\text{e}^{\frac{1-a}{a}}\geq \text{e}^{-\ln(a)}=\frac{1}{a}.
	\end{align}
	Using the Maclaurin series expansion, we obtain
	\begin{align}
		1+\frac{1-a}{a}+\frac{1}{2!}\left(\frac{1-a}{a}\right)^2+\frac{1}{3!}\left(\frac{1-a}{a}\right)^3+\ldots\geq \frac{1}{a}
	\end{align}
	which is obviously true. Hence, the value of first derivative is positive for every natural $l$, which means that the function on the left hand side of (\ref{eq:to_prove_final1}) is monotonically increasing.    As a consequence, it is also non-negative, since for $l=1$ it is equal to zero. This completes the verification of the constraints for $k<1/c_+$. 
	\item[III.]
	Now, let $k>1/c_++1$. Set $\tilde{l}+1=k-1/c_+>0$. Analogously to the previous case, we may rewrite Eq. (\ref{eq:constraints_k}) in the following form:
	\begin{align}
		\label{eq:constraint_to_prove_2}
		1-a^{\tilde{l}+1}-(1-a)(\tilde{l}+1)\left(\frac{n-1}{n}\right)^{\tilde{l}}\geq 0.
	\end{align}
	Due to Eq. (\ref{eq:p_min<1/n}), we obtain 
	\begin{align}
		a^{\tilde{l}}\geq \left(\frac{n-1}{n}\right)^{\tilde{l}},
	\end{align}
	which implies that to prove Eq. (\ref{eq:constraint_to_prove_2}), it is enough to show that
	\begin{align}
		\label{eq:to_prove_final2}
		1-a^{\tilde{l}+1}-(1-a)(\tilde{l}+1)a^{\tilde{l}}\geq 0.
	\end{align}
	We obtain
\begingroup\makeatletter\def\f@size{9.5}\check@mathfonts	
	\begin{align}
		\begin{aligned}
			\frac{d}{d\tilde{l}}\left(1-a^{\tilde{l}+1}-(1-a)(\tilde{l}+1)a^{\tilde{l}}\right)
			&=-a^{\tilde{l}+1}\ln(a)-(1-a)a^{\tilde{l}}-(1-a)(\tilde{l}+1)a^{\tilde{l}}\ln(a)\\
			&\geq a^{\tilde{l}} \left(-a\ln(a)-(1-a)-2(1-a)\ln(a)\right).
		\end{aligned}
	\end{align} 
	\endgroup
	The derivative is positive, i.e.
	\begin{align}
		-a\ln(a)-(1-a)-2(1-a)\ln(a)\geq 0
	\end{align}
	if
	\begin{align}
		\ln\left(\frac{1}{a}\right)\geq \frac{1-a}{2-a}.
	\end{align}
	Note that it is enough to verify that
	\begin{align}
		\frac{1}{a}\geq \text{e}^{1-a}
	\end{align}
	and this is easily verified by the series expansions of $1/(1-x)$ and $\exp(x)$.
	Since, we established positivity of the first derivative for every natural $l$, we know that the function on the left hand side of Eq. (\ref{eq:to_prove_final2}) is increasing. As a consequence, the function is also non-negative, which follows from the result for $l=1$, namely that $1-a^{2}-2(1-a)a = (1-a)^2 \geq 0$. 
\end{itemize}
\end{proof}

\section*{Acknowledgements}
The authors would like to thank Roger Colbeck, Renato Renner, Christopher Portmann, and Gilles P\"utz for useful discussions. 
This work was supported by the EU grant RAQUEL, the ERC AdG QOLAPS and by the John Templeton Foundation. The opinions expressed in this publication are those of the authors and do not necessarily reflect the views of the John Templeton Foundation. 
	The work of K. Horodecki and M. Stankiewicz was supported by National Science Centre grant Sonata Bis 5 no. 2015/18/E/ST2/00327.
	The work of M. Paw\l{}owski was supported by the Foundation for Polish Science (FNP) grant First TEAM/2016-1/5. 
	The work of H. Wojew\'odka was supported by the Foundation for Polish Science (FNP).
	This paper was presented in part at the Randomness in Quantum Physics and Beyond conference, Institute of Photonic Sciences, Barcelona, Spain, May 2015,
	at the 13th Central European Quantum Information Processing Workshop, Valtice, Czech Republic, June 2016,
	at the 3rd Seefeld Workshop on Quantum Information, Seefeld, Austria, June 2016,
	at the 6th International Conference on Quantum Cryptography (QCrypt), Washington, DC, September 2016,
	and at the 11th Conference on the Theory of Quantum Computation, Communication and Cryptography, Berlin, Germany, September 2016.

\bibliography{bibliography2}

\begin{thebibliography}{10}

\bibitem{rotem}
Rotem Arnon-Friedman and Amnon Ta-Shma.
\newblock Limits of privacy amplification against nonsignaling memory attacks.
\newblock {\em Physical Review A}, 86(6), December 2012.

\bibitem{brghhh}
Fernando G. S.~L. Brand{\~{a}}o, Ravishankar Ramanathan, Andrzej Grudka, Karol
  Horodecki, Micha{\l} Horodecki, Pawe{\l} Horodecki, Tomasz Szarek, and Hanna
  Wojew{\'{o}}dka.
\newblock Realistic noise-tolerant randomness amplification using finite number
  of devices.
\newblock {\em Nature Communications}, 7:11345, April 2016.

\bibitem{braunstein_caves}
Samuel~L. Braunstein and Carlton~M. Caves.
\newblock Wringing out better bell inequalities.
\newblock {\em Nuclear Physics B - Proceedings Supplements}, 6:211--221, March
  1989.

\bibitem{block_min}
Ben Chor and Oded Goldreich.
\newblock Unbiased bits from sources of weak randomness and probabilistic
  communication complexity.
\newblock {\em SIAM J. Comput.}, 17(2):230--261, April 1988.

\bibitem{Chung-Shi-Wu}
Kai-Min {Chung}, Yaoyun {Shi}, and Xiaodi {Wu}.
\newblock {Physical Randomness Extractors: Generating Random Numbers with
  Minimal Assumptions}.
\newblock unpublished, February 2014.

\bibitem{chungNew}
Kai-Min {Chung}, Yaoyun {Shi}, and Xiaodi {Wu}.
\newblock {General Randomness Amplification with Non-signaling Security}.
\newblock unpublished, 2016.

\bibitem{cr}
Roger Colbeck and Renato Renner.
\newblock Free randomness can be amplified.
\newblock {\em Nature Physics}, 8(6):450--454, May 2012.

\bibitem{acin}
Rodrigo Gallego, Lluis Masanes, Gonzalo De~La Torre, Chirag Dhara, Leandro
  Aolita, and Antonio Ac{\'{\i}}n.
\newblock Full randomness from arbitrarily deterministic events.
\newblock {\em Nature Communications}, 4, October 2013.

\bibitem{gms-pseudo-telepathy}
Nicolas Gisin, Andr{\'{e}}~Allan M{\'{e}}thot, and Valerio Scarani.
\newblock Pseudo-telepathy: input cardinality and bell-type inequalities.
\newblock {\em International Journal of Quantum Information}, 05(04):525--534,
  August 2007.

\bibitem{paradox}
Daniel~M. Greenberger, Michael~A. Horne, and Anton Zeilinger.
\newblock Going beyond bell's theorem.
\newblock In {\em Bell's Theorem, Quantum Theory and Conceptions of the
  Universe}, pages 69--72. Springer Netherlands, 1989.

\bibitem{ghhhpr}
Andrzej Grudka, Karol Horodecki, Micha{\l} Horodecki, Pawe{\l} Horodecki,
  Marcin Paw{\l}owski, and Ravishankar Ramanathan.
\newblock Free randomness amplification using bipartite chain correlations.
\newblock {\em Physical Review A}, 90(3), September 2014.

\bibitem{gthb-graph-states}
Otfried G\"{u}hne, G{\'{e}}za T{\'{o}}th, Philipp Hyllus, and Hans~J. Briegel.
\newblock Bell inequalities for graph states.
\newblock {\em Physical Review Letters}, 95(12), September 2005.

\bibitem{maciek}
Karol {Horodecki}, Micha{\l} {Horodecki}, Pawe{\l} {Horodecki}, Ravishankar
  {Ramanathan}, Maciej {Stankiewicz}, and Hanna {Wojew{\'{o}}dka}.
\newblock {Randomness amplification using independent devices arbitrarily
  correlated with the Santha-Vasirani source}.
\newblock unpublished, May 2017.

\bibitem{jm}
Nick~S. Jones and Llu{\'{\i}}s Masanes.
\newblock Interconversion of nonlocal correlations.
\newblock {\em Physical Review A}, 72(5), November 2005.

\bibitem{KRS}
Robert Konig, Renato Renner, and Christian Schaffner.
\newblock The operational meaning of min- and max-entropy.
\newblock {\em {IEEE} Transactions on Information Theory}, 55(9):4337--4347,
  September 2009.

\bibitem{marcin}
Piotr Mironowicz, Rodrigo Gallego, and Marcin Paw{\l}owski.
\newblock Robust amplification of santha-vazirani sources with three devices.
\newblock {\em Physical Review A}, 91(3), March 2015.

\bibitem{ravi}
Ravishankar Ramanathan, Fernando~G.{\hspace{0.167em}}S.{\hspace{0.167em}}L.
  Brand{\~{a}}o, Karol Horodecki, Micha{\l} Horodecki, Pawe{\l} Horodecki, and
  Hanna Wojew{\'{o}}dka.
\newblock Randomness amplification under minimal fundamental assumptions on the
  devices.
\newblock {\em Physical Review Letters}, 117(23), November 2016.

\bibitem{tuziemski}
Ravishankar Ramanathan, Jan Tuziemski, Micha{\l} Horodecki, and Pawe{\l}
  Horodecki.
\newblock No quantum realization of extremal no-signaling boxes.
\newblock {\em Physical Review Letters}, 117(5), July 2016.

\bibitem{sv}
Miklos Santha and Umesh~V. Vazirani.
\newblock Generating quasi-random sequences from semi-random sources.
\newblock {\em Journal of Computer and System Sciences}, 33(1):75--87, August
  1986.

\end{thebibliography}
\bibliographystyle{plain}

\end{document}